\documentclass[12pt]{article}

\usepackage{amsmath,amsthm,amsfonts,amssymb,mathtools}
\usepackage{float}
\usepackage[svgnames,dvipsnames]{xcolor}
\usepackage{enumerate}
\usepackage{graphicx}
\usepackage{epstopdf}
\usepackage{tikz-cd}
\usepackage[all]{xy}

 \newcommand{\hooklongrightarrow}{\lhook\joinrel\longrightarrow} 

\definecolor{dullmagenta}{RGB}{102,0,102}
\def\col{dullmagenta} 

\usepackage[
    colorlinks=true,       
    linkcolor=\col,          
    citecolor=\col,      
    filecolor=\col,       
    urlcolor=\col,            
    ]{hyperref}

\makeatletter 
\renewcommand\paragraph{\@startsection{paragraph}{4}{\z@}%
                                      {\parskip}
                                      {-1em}%
                                      {\normalfont\normalsize\bfseries}}
\makeatother

\parskip=\medskipamount
\makeatletter 
\def\@endtheorem{\endtrivlist}
\makeatother

\makeatletter
\DeclareFontFamily{OMX}{MnSymbolE}{}
\DeclareSymbolFont{MnLargeSymbols}{OMX}{MnSymbolE}{m}{n}
\SetSymbolFont{MnLargeSymbols}{bold}{OMX}{MnSymbolE}{b}{n}
\DeclareFontShape{OMX}{MnSymbolE}{m}{n}{
    <-6>  MnSymbolE5
   <6-7>  MnSymbolE6
   <7-8>  MnSymbolE7
   <8-9>  MnSymbolE8
   <9-10> MnSymbolE9
  <10-12> MnSymbolE10
  <12->   MnSymbolE12
}{}
\DeclareFontShape{OMX}{MnSymbolE}{b}{n}{
    <-6>  MnSymbolE-Bold5
   <6-7>  MnSymbolE-Bold6
   <7-8>  MnSymbolE-Bold7
   <8-9>  MnSymbolE-Bold8
   <9-10> MnSymbolE-Bold9
  <10-12> MnSymbolE-Bold10
  <12->   MnSymbolE-Bold12
}{}

\let\llangle\@undefined
\let\rrangle\@undefined
\DeclareMathDelimiter{\llangle}{\mathopen}%
                     {MnLargeSymbols}{'164}{MnLargeSymbols}{'164}
\DeclareMathDelimiter{\rrangle}{\mathclose}%
                     {MnLargeSymbols}{'171}{MnLargeSymbols}{'171}
\makeatother

\theoremstyle{plain}
\newtheorem{theorem}{Theorem}

\newtheorem{proposition}{Proposition}

\theoremstyle{definition}
\newtheorem{remark}{Remark}
\newtheorem{example}{Example}
\newtheorem{definition}{Definition}

\newcommand{\F}{\mathcal{F}}
\newcommand{\B}{\mathcal{B}}
\newcommand{\R}{\mathbb{R}}
\newcommand{\A}{\mathcal{A}}

\DeclareMathOperator{\Img}{Im}
\DeclareMathOperator{\rank}{rank}

\newcommand{\st}{\;\ifnum\currentgrouptype=16 \middle\fi|\;}

\begin{document}
\title{Morse families and Dirac systems}

\author{M.\ Barbero-Li\~n\'an\textsuperscript{\textsection,$\ddagger$}, H.\ Cendra\textsuperscript{$\dagger$}, E.\ Garc\'{\i}a-Tora\~{n}o Andr\'{e}s\textsuperscript{$\dagger$},  and \\ D.\ Mart\'{\i}n de Diego\textsuperscript{$\ddagger$}
\\[2mm]
{\small \textsection\, Departamento de Matem\'atica Aplicada, Universidad Polit\'ecnica de Madrid,} \\
{\small Av. Juan de Herrera 4, 28040 Madrid, Spain}
\\[2mm]
{\small $\dagger$ Departamento de Matem\'atica,
Universidad Nacional del Sur, CONICET,}  \\
{\small  Av.\ Alem 1253, 8000 Bah\'ia Blanca, Argentina}
\\[2mm]
{\small $\ddagger$  Instituto de Ciencias Matem\'aticas (CSIC-UAM-UC3M-UCM),} \\
{\small C/Nicol\'as Cabrera 13-15, 28049 Madrid, Spain}}

\date{}

\maketitle

\begin{abstract}
Dirac structures and Morse families are used to obtain a geometric formalism that unifies most of the scenarios in mechanics (constrained calculus, nonholonomic systems, optimal control theory, higher-order mechanics, etc.), as the examples in the paper show. This approach generalizes the previous results on Dirac structures associated with Lagrangian submanifolds. An integrability algorithm in the sense of Mendela, Marmo and Tulczyjew is described for the generalized Dirac dynamical systems under study to determine the set where the implicit differential equations have solutions.
\end{abstract}

\section{Introduction}\label{sec:Intro}

Dirac structures were introduced in~\cite{CouWei,Courant} as a unified approach to both symplectic and Poisson geometries. One of the motivations behind the definition was the study of constrained systems, including the celebrated case of the  constrained bracket induced by a degenerate Lagrangian (which was first studied by Dirac~\cite{Dirac50,Dirac}, and after whom Dirac structures are named). The infinite-dimensional analog was introduced in~\cite{Dorfman2} in the context of integrable evolution equations. 

Dirac structures were soon employed to describe many situations of interest in mechanics and mathematical physics. In particular the idea of using a Dirac structure $D\subset TM\oplus T^*M$, where $M=T^*Q$ and $D$ is induced by the canonical symplectic structure, and a Hamiltonian function $H\colon M\to\R$ which represents the energy $E=H$ to write an implicit Hamiltonian system of the form
\begin{equation}\label{eq:intro1}
 \dot x\oplus dE(x)\in D_x
\end{equation}
is already found in~\cite{VdSM2,VdSM3}. This has been the cornerstone of the development of a geometric theory of Port-Hamiltonian systems~\cite{DVdS}; we refer the interested reader to~\cite{VdS-Book} for a survey and a comprehensive list of references. Building on the notion of implicit Hamiltonian system, the case of Lagrangian systems (possibly degenerate) was described in \cite{2006YoshiMarsdenI,2006YoshiMarsdenII} using the so-called ``Dirac differential'' $\mathfrak{D} L\colon TQ\to T^*T^*Q$ of the Lagrangian $L\colon TQ\to\R$. In a nutshell, the Dirac differential combines Tulczyjew's diffeomorphisms and the differential of $L$ to define a subset of $T^*T^*Q$. It is then proved that the implicit Lagrangian system
\begin{equation}\label{eq:intro2}
X\oplus \mathfrak{D} L\in D_{\omega_Q},
\end{equation}
where $X\colon TQ\oplus T^*Q\to TT^*Q$ is a \emph{partial vector field} and $D_{\omega_Q}\subset TT^*Q\oplus T^*T^*Q$ is the Dirac structure induced by the canonical symplectic form $\omega_Q$ on $T^*Q$, leads to the equations of the Lagrangian system written in an implicit form. One of the virtues of~\eqref{eq:intro2} is that one might modify the Dirac structure $D_{\omega_Q}$ to account for a nonholonomic distribution $\Delta_Q\subset TQ$, and produce the standard nonholonomic equations (again, the equations are obtained in an implicit way). One can modify this approach to include more general situations such as vakonomic mechanics~\cite{2015JiYo}.

Dirac systems, i.e. system of the form~\eqref{eq:intro1}, also include Lagrangian systems given by $L\colon TQ\to\R$. In this case, one enlarges the phase space and works on the Pontryagin bundle $M=TQ\oplus T^*Q$, endowed with a suitable Dirac structure, and the energy $E\colon M\to \R$ given by $E(q,v,p)=pv-L(q,v)$.  Within this framework, there is no need for neither a Dirac differential operator nor the notion of a partial vector field as in~\cite{2006YoshiMarsdenI,2006YoshiMarsdenII}The reader can take a look at~\cite{CEF_DiracConstraints} for the more elaborated examples of nonholonomic systems and LC-circuits.

Another unified approach to Dirac systems, but based on the more general notion of Dirac algebroids, is found in~\cite{Dirac_Algebroids}. Using a somehow similar approach to that of \cite{2006YoshiMarsdenI}, the authors develop a formalism which includes, among others, mechanics on algebroids, non-autonomous systems and vakonomic and nonholonomic mechanics.

In this paper, we wish to take an alternative point of view and generalize Dirac systems in such a way that dynamics is defined by means of a Lagrangian submanifold of the phase space. We will use Morse families to generate those Lagrangian submanifolds, an approach that has already been applied in the realm of optimal control problems~\cite{2012BaIgleMar}. With this notion of \emph{generalized Dirac system} we are able to recover many examples in the literature in a unified intrinsic formalism, including nonholonomic and vakonomic mechanics, optimal control problems and constrained problems on linear almost Poisson manifolds.

The paper starts with a brief review of the relevant definitions in Dirac geometry (Section~\ref{sec:Background}). In Section~\ref{sec:Morse} the basic construction of Lagrangian submanifolds out of Morse families is explained (we follow closely the exposition in~\cite{LiMarle}), and we define the notion of generalized Dirac system that we will consider in this paper. Theorem~\ref{Thm:main} relates this approach with the standard notion of Dirac system. Section~\ref{sec:Examples} is devoted to a number of examples and in Section~\ref{sec:algorithm} we discuss the application of the integrability algorithm in~\cite{1995MMT} to the notion of generalized Dirac system and discuss the preservation of the presymplectic structure by the dynamics under the assumption of the integrability of the Dirac structure. The paper ends with a section devoted to future work. The final appendix contains details about the relation between the Dirac structures of interest to us and some natural maps in mechanics.

\paragraph{Acknowledgements.} The authors acknowledge the financial support from the Spanish government through the network research grant MTM2015-69124-REDT. MB and DMdD acknowledge financial support from the Spanish Ministry of Economy and Competitiveness, through the   research grants MTM2013-42870-P, MTM2016-76702-P and ``Severo Ochoa Programme for Centers of Excellence'' in R\&D (SEV-2015-0554). HC and EGTA thank the CONICET for financial support. HC and EGTA also acknowledge funding from ANPCyT (BID PICT 2013 1302, Argentina) and from Universidad Nacional del Sur (PGI 24/L098). EGTA  is grateful to the ICMAT (Madrid) for its hospitality during the visits which made this work possible.

\section{Preliminaries on Dirac structures}
\label{sec:Background}

Consider a finite dimensional vector space $U$, and let $U^*$ be its dual. We endow the vector space $U\oplus U^*$ with the following symmetric non-degenerate pairing:
\[
\llangle (u_1,\alpha_1),(u_2,\alpha_2) \rrangle=\langle \alpha_1, u_2\rangle + \langle \alpha_2,u_1\rangle\,,
\]
where $(u_1,\alpha_1),(u_1,\alpha_2)\in U\oplus U^*$. A linear Dirac structure on $U\oplus U^*$ is a Lagrangian subspace $D\subset U\oplus U^*$ or, in other words, $D$ satisfies $D=D^\perp$ where the orthogonal is taken relative to the pairing  $\llangle\cdot,\cdot\rrangle$. It is easy to show that a Dirac structure is characterized by the following two conditions:  $\llangle (u_1,\alpha_1),(u_2,\alpha_2) \rrangle=0$ for all $(u_1,\alpha_1),(u_2,\alpha_2)\in U\oplus U^*$, and ${\rm dim} \;D= {\rm dim}\; U$. 

\begin{definition}
A \emph{Dirac structure} on a manifold $M$ is a subbundle $D\subset TM\oplus T^*M$ such that, for each $m\in M$, $D_m\subset T_mM\oplus T_m^*M$ is a linear Dirac structure.
\end{definition}
We remark that we do not require any integrability condition in this definition (this is the convention in e.g.~\cite{DVdS,2006YoshiMarsdenI}). The reader can find more details about the geometric meaning of the integrability condition on~\cite{Courant}. Two basic examples of Dirac structures are the following:
\begin{enumerate}[(i)]
 \item If $\omega$ is a 2-form on $M$, the musical isomorphism $\omega^\flat\colon TM\to T^*M$ is defnined by $v\mapsto \omega^\flat(v)=\omega(v,\cdot)$. Its graph defines a Dirac structure that we denote $D_\omega$:
 \[
 D_\omega=\{(v,\omega^\flat(v))\st v\in TM\}\subset TM\oplus T^*M.
 \]
 \item Let $\Lambda$ be a Poisson bivector on $M$, and let us denote by $\sharp_\Lambda\colon T^*M\to TM$ the isomorphism given by $\langle\beta,\sharp_\Lambda(\alpha)\rangle=\Lambda(\beta,\alpha)$, for each $\alpha,\beta\in T^*M$. The graph of $\sharp_\Lambda$ defines a Dirac structure:
 \[
D_\Lambda=\{(\sharp_\Lambda(\alpha),\alpha)\st \alpha\in T^*M\}\subset TM\oplus T^*M.
 \]
\end{enumerate}
More general Dirac structures can be obtained restricting forms or bivectors to distributions or codistributions, respectively. We refer to~\cite{DVdS,2006YoshiMarsdenI}) for more details and examples. We remark that, in general, Dirac structures are not given by graphs of forms or bivectors, see e.g.~\cite{Burs}.

Within the framework of Dirac manifolds (manifolds equipped with a Dirac structure), there are operations of ``backward'' and ``forward'' which extend the usual notions of pull-back of a 2-form and push-forward of a bivector (see~\cite{RadkoBursztyn,Burs}). We start with the construction in the case of vector spaces. Let $\varphi\colon U\to V$ be a linear map between the vector spaces $U$ and $V$. The following holds:
\begin{enumerate}[1)]
 \item If $D_V$ is a linear Dirac structure on $V$, then
 \begin{equation}\label{eq:backward_vectorspace}
\B_\varphi(D_V)=\{(u,\varphi^*v^*)\in U\oplus U^* \st u\in U, \; v^*\in V^*, \; (\varphi u,v^*)\in D_V\}
  \end{equation}
is a linear Dirac structure on $U$ that we call \emph{backward of $D_U$ by $\varphi$}.
\item If $D_U$ is a linear Dirac structure on $U$, then
\begin{equation}\label{eq:forward_vectorspace}
\F_\varphi(D_U)=\{(\varphi u,v^*)\in V\oplus V^* \st u\in U, \; v^*\in V^*, \; (u,\varphi^*v^*)\in D_U\}
\end{equation}
is a linear Dirac structure on $V$ that we call \emph{forward of $D_U$ by $\varphi$}.
\end{enumerate}

These operations can be extended pointwise to the case of Dirac manifolds, although one needs some regularity condition to assure that the resulting distributions are smooth (and hence, define Dirac structures). In more detail, the construction of the backward and the forward for Dirac manifolds are as follows. Let $f\colon M\to N$ be a smooth map, then:

\begin{enumerate}[1)]
 \item Let $D_N\subset TN\oplus T^*N$ be a Dirac structure on $N$. For each $m\in M$ the map $T_mf\colon T_mM\to T_{f(m)}N$ can be used to define pointwise the backward image of $(D_N)_{f(m)}$ by $T_mf$ in the sense of~\eqref{eq:backward_vectorspace}. This defines a Lagrangian distribution of $TM\oplus T^*M$. When it is a subbundle (i.e. the distribution is regular), then it becomes a Dirac structure $D_M$ on $M$ that we call the \emph{backward of $D_N$ by $f$}, and denote it by $D_M\equiv\B_f(D_N)$.

\item Let $D_M\subset TM\oplus T^*M$ be a Dirac structure on $M$ which is $f$-invariant, meaning that
\[
\F_{(T_mf)}\big((D_M)_m\big)=\F_{(T_{m'}f)}\big((D_M)_{m'}\big),\qquad \text{whenever } f(m)=f(m'),
\]
with the forward $\F_{(T_mf)}$ as in~\eqref{eq:forward_vectorspace}. Then, similarly to the backward case, one can use the forward construction to define pointwise a Lagrangian distribution of $TN\oplus T^*N$. Whenever this distribution becomes smooth, it defines a Dirac structure on $N$ that we call  the \emph{forward of $D_M$ by $f$}, and denote it by $D_N\equiv\F_f(D_M)$.
\end{enumerate}

Some regularity conditions which guarantee that these constructions define Dirac structures can be found in~\cite{Burs,2012CenRaYo}. The forward and the backward of a Dirac structure have the following functorial property: if $f\colon M\to N$ and $g\colon N\to Q$ are smooth maps, then
\[
\B_{(g\circ f)}=\B_f\circ\B_g,\qquad \F_{(g\circ f)}=\F_g\circ\F_f.
\]

\section{Morse families and generalized Dirac systems}
\label{sec:Morse}

The notion of a Morse family or
phase function was introduced in~\cite{Hormander}.  Here we just give some key definitions and essential results that we will need later.  More details can be found in~\cite{AbMa,1989LeRo,GuiStern,LiMarle,Weinstein} and references therein. In particular, we will follow~\cite{LiMarle}.

\paragraph{Morse families and Lagrangian submanifolds.} A simple but important example of a Lagrangian submanifold of $T^* Q$ is provided by the image $\alpha (Q)$ of a closed 1-form $\alpha$ on $Q$. More general Lagrangian submanifolds of cotangent bundles can be represented as a certain quotient of images of 1-forms, as we will explain next following~\cite{LiMarle}. Recall that if $\pi\colon M \rightarrow N$ is a submersion the conormal bundle is
\[
(\ker\, T \pi)^\circ =\left\{\alpha\in T^*M\; |\; \langle \alpha, v\rangle=0, \hbox{ for all } v\in \ker\, T_{\pi_M(\alpha)}\pi\right\}\subset T^*M.
\]

\noindent Define the vector bundle morphism $\mathfrak{j}_\pi\colon (\ker \, T \pi)^\circ \rightarrow T^*N$  by
\begin{equation}\label{eq:jpi}
 \langle \mathfrak{j}_\pi(\eta), T\pi (v) \rangle=\langle \eta,v\rangle
\end{equation}
for all $\eta\in (\ker \, T \pi)^\circ$ and for all $v\in T_{\pi (\eta)} M$. It is easy to check that this application is well defined and that the following diagram is commutative

\begin{figure}[H]
\centering
\includegraphics{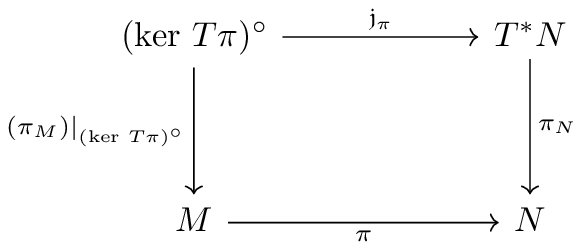}
\end{figure}

\noindent One can also show using the definition of ${\mathfrak j}_{\pi}$~\eqref{eq:jpi} that the restriction of ${\mathfrak j}_{\pi}$  to each fiber is an isomorphism from that fiber to the corresponding fiber of $T^*N$. The map ${\mathfrak j}_{\pi}$
can be locally described as follows.
Let $(q^i, y^a)$ be fibered coordinates  on $M$ and $(q^i)$ on $N$ such that $\pi(q^i, y^a)=(q^i)$ which induce coordinates $(q^i, y^a, p_i, z_a)$ on $T^*M$ and $(q^i, p_i)$ on $T^*N$. The subbundle $(\ker T\pi)^\circ$ is locally described by coordinates $(q^i, y^a, p_i,z_a=0)$, and
\[
\begin{array}{rrcl}
{\mathfrak j}_{\pi}:& (\ker T\pi)^\circ&\longrightarrow& T^*N,\\
                             & (q^i, y^a, p_i,z_a=0)&\longmapsto&(q^i, p_i).
\end{array}
\]
We write $\pi^* \colon T_{\pi(\eta)}^*N \rightarrow T_{\eta}^*M$ for the pullback of $\pi$. Note that
\begin{equation*}
\pi^*\circ {\mathfrak j}_{\pi}={\rm id}_{(\ker T\pi)^\circ}\, .
\end{equation*}

\begin{definition}\label{def:MorseFamily}\normalfont
Let $\pi\colon M \rightarrow N$ be a submersion of a differentiable manifold $M$ onto a differentiable manifold $N$. Let $E\colon M\rightarrow \mathbb{R}$ be a differentiable function. The function $E$ is called a \textit{Morse family over the submersion $\pi$} if the image of the differential of $E$, $ {\rm d}E(M)\subset T^*M$, and the conormal bundle are transverse in $T^*M$.
\end{definition}
\noindent Recall that by definition ${\rm d}E(M)$ and $(\ker\, T\pi)^\circ$ are transverse in $T^*M$, denoted by ${\rm d}E(M)\pitchfork (\ker\, T\pi)^\circ$, if
\begin{equation*}\forall\, \alpha\in (\ker \, T \pi)^\circ \cap  {\rm d}E(M)  \subset T^*M,   \qquad T_\alpha( {\rm d}E(M))+T_\alpha (\ker \, T \pi)^\circ=T_\alpha (T^*M).
\end{equation*}
In coordinates adapted to the fibration, if the submersion $\pi: M\rightarrow N$ is expressed by $\pi(q^i, y^a)=(q^i)$, then the condition for $E: M\rightarrow \mathbb{R}$ to be a Morse family is that the matrix
\begin{equation}\label{eq:Morsematrix}
\left(
\frac{\partial^2 E}{\partial q^i \partial y^b}, \frac{\partial^2 E}{\partial y^a \partial y^b}
\right)
\end{equation}
has maximal rank for all $(q^i, y^a)$ satisfying $\displaystyle{\frac{\partial E}{\partial y^a}=0}$.

Observe that as a consequence of Definition~\ref{def:MorseFamily} the submanifold $(\ker \, T \pi)^\circ \cap  {\rm d}E(M)$ is isotropic in $(T^*M, \omega_M)$, since it is contained in the Lagrangian submanifold ${\rm d}E(M)$. A computation shows that $\dim \left((\ker \, T \pi)^\circ \cap  {\rm d}E(M)\right)=\dim N$ (see~\cite{LiMarle}). We remark that the restriction of the canonical symplectic form $\omega_M$ to $(\ker\, T\pi)^\circ$ is equal to $\pi^*\omega_N$, locally ${\rm d}q^i\wedge {\rm d}p_i$. From here we derive the following useful result (see also~\cite[Appendix 7, Proposition 1.12]{LiMarle} or Chapter 4 in \cite{GuiStern}).

\begin{proposition}\label{prop:img:j:lagrangian}
Let $E\colon M\rightarrow \mathbb{R}$ be a Morse family. The restriction of the morphism $\mathfrak{j}_\pi\colon (\ker \, T \pi)^\circ \rightarrow T^*N$ to the isotropic submanifold $(\ker \, T \pi)^\circ \cap {\rm d}E(M)$ is a Lagrangian immersion of $(\ker \, T \pi)^\circ \cap {\rm d}E(M)$ in $(T^*N, \omega_N)$. This Lagrangian immersion is said to be generated by the Morse family $E$.
\label{Prop:MorseFamily}
\end{proposition}
\noindent We will denote by
\[
S_E={\mathfrak j}_{\pi}\left({\rm d}E(M)\cap (\ker\, T\pi)^\circ\right)
\]
the immersed Lagrangian submanifold in the above proposition. Observe that in general $S_E$ is not horizontal, that is, it is not transverse to the fibers of the canonical cotangent projection $\pi_N$, and consequently, it is not the image of the differential of a function on $N$.

\paragraph{Weak Morse families.} We will now describe a certain extension of the basic results on Morse families discussed in the previous paragraphs. Namely we will study cases where we have the weaker condition that $\left(\ker\, T \pi \right)^\circ \cap {\rm d} E (M)$ is a submanifold and for all $\alpha\in \left(\ker\, T \pi \right)^\circ \cap {\rm d} E (M)$ we have that
\[
T_{\alpha}\big( \left(\ker\, T \pi \right)^\circ \cap {\rm d} E (M)\big)=
T_{\alpha}\left(\ker\, T \pi \right)^\circ\cap T_{\alpha}{\rm d} E (M)\,.
\]
That is, we are assuming that $\left(\ker\, T \pi \right)^\circ$ and ${\rm d} E (M)$ are weakly transverse (this is sometimes called clean intersection).
Under the condition of weak transversality we have that
\[
\left.{\mathfrak j}_{\pi} \right |_{ \left(\ker\, T \pi\right)^\circ \cap {\rm d} E (M)}: \left(\ker\, T \pi\right)^\circ \cap {\rm d} E (M)\rightarrow T^*N
\]
is of constant rank (a subimmersion). With this assumption,  ${\mathfrak j}_{\pi} (\left(\ker\, T \pi\right)^\circ \cap {\rm d} E (M))$
is an immersed Lagrangian submanifold of $T^* N$ which may include multiple points (see~\cite{LiMarle}).

The local criteria we will be using for weak transversality of ${\rm d}E(M)$ is the following. The submanifold $\left(\ker\, T \pi\right)^\circ$ is described locally by $z_a=0$, $a=1,\dots,k$ (where $k={\dim M-\dim N}$). Let $z=(z_1,\dots,z_k)$, which is a $\mathbb{R}^k$-valued function. We define the set
\[
M_0=\{m\in M\st {\rm d}E(m)\in\left(\ker\, T \pi\right)^\circ \}\subset M.
\]
Locally, $M_0$ coincides with the set $\left(z\circ {\rm d}E\right)^{-1}(0)$. The differential of the map $z\circ {\rm d}E$ is the matrix~\eqref{eq:Morsematrix}. Therefore, if its rank is constant, using the constant rank theorem the set $\left(z\circ {\rm d}E\right)^{-1}(0)$ defines an embedded submanifold whose tangent space is $\Img {\rm d}E\cap \ker\, z$. To sum up, if the matrix in~\eqref{eq:Morsematrix} has constant rank (not necessarily maximum), ${\rm d} E(M_0)= \left(\ker\, T \pi\right)^\circ \cap {\rm d} E (M)$ is an embedded submanifold of $T^*M$ whose tangent space satisfies the clean intersection condition.

\begin{definition}\normalfont
With the notation introduced above, $E$ will be called a \emph{weak Morse family} if
$\left(\ker\, T \pi \right)^\circ$ and ${\rm d}E (M)$ are weakly transverse.\end{definition}

\begin{example}The following simple example of weak Morse family is of interest in this paper. Given the submersion $\pi: M\rightarrow N$ and a function $f: N\rightarrow {\mathbb R}$, take $E=\pi^*f: M\rightarrow {\mathbb R}$. Obviously $E$ is not a Morse family (the rank of~\eqref{eq:Morsematrix} is zero), but it is a weak Morse family and moreover
\[
{\mathfrak j}_{\pi} (\left(\ker\, T \pi\right)^\circ \cap {\rm d} E (M))=\hbox{Im } {\rm d} f\,,
\]
which is  a Lagrangian submanifold of $(T^*N, \omega_N)$.
\end{example}
\vspace{.3cm}

\paragraph{Dirac systems and Dirac systems over a weak Morse family.} Assume that $M$ is endowed with a Dirac structure $D_M$, and let $E\colon M\to \mathbb{R}$ be a given \emph{energy} function. We  consider the following implicit dynamical system: for a curve $\gamma\colon I\to M$ (where $I\subset\mathbb{R}$ is an interval), we say that $\gamma$ is a solution of the \emph{Dirac system $(D_M,{\rm d}E)$} if
\begin{equation}\label{eq:DiracDynamics}
\dot \gamma(t) \oplus {\rm d}E\left(\gamma(t)\right)\in (D_M)_{\gamma(t)}\quad  \mbox{ for all }\, t\in I.
\end{equation}
The system described by~\eqref{eq:DiracDynamics} is general enough to encompass a number of situations of interest in mathematical physics including  of course classical Lagrangian and Hamiltonian systems, but also nonholomic mechanics or electric LC circuits~\cite{2006YoshiMarsdenI,2006YoshiMarsdenII,CEF_DiracConstraints}.

We will now extend the definition of Dirac system to include more general Lagrangian submanifolds defined in terms of weak Morse families. As the examples in the next section show, this broader definition permits to describe more general dynamical systems in terms of a Dirac structure and a Lagrangian submanifold. With the notations used in this section, let $\pi\colon M\to N$ be a surjective submersion, $S_E\subset T^*N$ a Lagrangian submanifold induced by a weak Morse family $E : M  \rightarrow \mathbb{R}$, and  $D_N$ a Dirac structure on $N$. We look for curves $n(t)$ in $N$ which solve the following implicit dynamical system
\begin{equation}\label{eq:GDiracsystem}
\dot{n}(t)\oplus \mu_{n(t)}\in (D_N)_{n(t)}\quad \mbox{ for all }\, t\in I,
\end{equation}
where $\mu_{n(t)}\in (S_E)_{n(t)}$, i.e. $\mu_{n(t)}\in S_E$ and $\pi_N(\mu_{n(t)})=n(t)$ (recall that $\pi_N\colon T^*N\to N$ is the canonical projection of the cotangent bundle). We will say that the dynamical system~\eqref{eq:GDiracsystem} is a \emph{Dirac system over $E$}. We will also refer to~\eqref{eq:GDiracsystem} as the \emph{(generalized) Dirac system $(D_N, S_E)$}.

The solution curves of the Dirac system over a Morse family can be alternatively described as a projection of solution curves of the Dirac system $(D_M, {\rm d}E)$, where $D_M=\B_\pi (D_N)$ is the backward of $D_N$ by $\pi$. Recall that the backward Dirac structure $D_M$ is well defined since $\pi$ is a submersion (see~\cite{Burs}). More precisely:

\begin{theorem}\label{Thm:main} Let $I$ be an interval of $\mathbb{R}$, $E: M\rightarrow {\mathbb R}$ a weak Morse family over the submersion $\pi\colon M\rightarrow N$, $D_N$ a Dirac structure over $N$, and $D_M=\B_\pi(D_N)$ the backward Dirac structure on $M$ induced by $D_N$.
\begin{enumerate}[i)]
\item
If $m: I\rightarrow M$ is a solution of the Dirac system determined by $(D_M, {\rm d}E)$, that is, for all $t$ in $I$
\[
\dot{m}(t)\oplus {\rm d}E(m(t))\in (D_M)_{m(t)}\,,
\]
then the curve $n: I\rightarrow N$ defined by $n=\pi\circ m$ is a solution of the Dirac system determined by $(D_N, S_E)$, that is, for all $t$ in $I$,
\[
\dot{n}(t)\oplus \mu_{n(t)}\in (D_N)_{n(t)}
\]
where $\mu_{n(t)}\in (S_E)_{n(t)}$.
\item Conversely, if $n: I \rightarrow N$ is a solution of the Dirac system determined by $(D_N, S_E)$,
then there exists a solution $m(t)$ of the Dirac system $(D_M, {\rm d}E)$ which projects onto $n$, i.e. $n=\pi\circ m$.
\end{enumerate}
\end{theorem}
\begin{proof}
Using the definition of $D_M=\B_\pi(D_N)$, $\dot{m}(t)\oplus {\rm d}E(m(t))\in (D_M)_{m(t)}$ implies ${\rm d}E(m(t))\in \left.(\ker\, T\pi)^\circ\right|_{m(t)}$. Therefore, $n=\pi\circ m$ satisfies
\[
{\mathfrak j}_{\pi}\big({\rm d}E(m(t))\big)\in (S_E)_{n(t)}
\]
and thus
\[
T\pi(\dot{m}(t))\oplus {\mathfrak j}_{\pi}({\rm d}E(m(t)))\in (D_N)_{n(t)}.
\]
This means that $n(t)$ solves the generalized Dirac system $(D_N, S_E)$.

The converse follows easily from the definitions.
\end{proof}

\remark Note that notion of Dirac system over a Morse family includes the case of a standard Dirac system $(D,{\rm d}E)$ as follows. If $D$ is a Dirac structure on $M$, one can consider the identity map on $M$, $1_M\colon M\to M$, and then the energy $E$ is obviously a Morse function for $1_M$. The Dirac system $(D,S_E)$ obtained coincides with the Dirac system $(D,{\rm d}E)$.

\section{Examples}\label{sec:Examples}

The purpose of this section is to show how the notion of generalized Dirac system covers many examples of interest in mechanics and control theory.

\subsection{Lagrangian and Hamiltonian Mechanics}\label{subsecc:LagraHamil}

We will describe this first example in some detail to clarify the notations and basic results of the previous sections.

\paragraph{Lagrangian mechanics.} Let $L: TQ\rightarrow {\mathbb R}$ be a Lagrangian, possibly degenerate, and consider the Dirac structure $D_{\omega_Q}$ on $T^*Q$ induced by the canonical symplectic $2$-form $\omega_Q$ on $T^*Q$. The coordinate expression of $D_{\omega_Q}$ reads:
\begin{equation*}\label{eq:D_omega}
D_{\omega_Q} =\left\{(q^i,p_i,\dot q^i,\dot p_i,\alpha_i,\beta^i)\,\mid\, \dot p_i +\alpha_i=0,\, \dot q^i-\beta^i=0\right\}.
\end{equation*}

Define the energy function $E: TQ\oplus T^*Q\rightarrow {\mathbb R}$ by
\[
E(q,v_q,\alpha_q)=\langle \alpha_q, v_q\rangle -L(v_q)\, .
\]
It is clear that $E$ is a Morse family for the projection $\text{pr}_2: TQ\oplus T^*Q\rightarrow T^*Q$, and therefore it generates the Lagrangian submanifold $S_E$ of $T^*T^*Q$. Indeed, the local expressions for the projection and the energy are $\text{pr}_2(q^i,v^i,p_i)= (q^i,p_i)$ and $E(q^i, v^i, p_i)=p_iv^i-L(q^i, v^i)$, and the matrix
\[
\left(
\frac{\partial^2 E}{\partial q^i \partial v^j}\,,\, \frac{\partial^2 E}{\partial p_i \partial v^j}\,,\, \frac{\partial^2 E}{\partial v^i \partial v^j}
\right)_{\{i,j\}}
=\left(-\frac{\partial^2 L}{\partial q^i \partial v^j}\,,\, \mathbb{I}\, ,\, -\frac{\partial^2 L}{\partial v^i \partial v^j}
\right)_{\{i,j\}}
\]
has maximal rank ($\mathbb{I}$ denotes de identity matrix of size equal to $\dim Q$).  We will also need the kernel of $T\text{pr}_2$ and its annihilator,
\begin{align*}
\ker\, (T\text{pr}_2)^{\phantom{\circ}}&=\left\{(q^i,v^i,p_i, \dot q^i,\dot v^i,\dot p_i)\,\mid\, \dot q^i=\dot p_i= 0 \right\}\subset T(TQ\oplus T^*Q)\,,\\
\left(\ker\, (T\text{pr}_2)\right)^\circ &=
\left\{(q^i,v^i,p_i, \alpha_i,\gamma_i,\beta^i)\st \gamma_i=0 \right\}\subset T^*(TQ\oplus T^*Q)\,.
\end{align*}
The local expression of the map $\mathfrak{j}_\pi\colon (\ker\, (T\text{pr}_2))^\circ\to T^*T^*Q$ is then
\[
\mathfrak{j}_\pi(q^i,v^i,p_i, \alpha_i,\gamma_i=0,\beta^i)=(q^i,p_i, \alpha_i,\beta^i),
\]
and from here the expression of $S_E=\mathfrak{j}_\pi\left({\rm d}E(M)\cap \big(\ker\, (T\text{pr}_2)\big)^\circ\right)$ follows:
\begin{align}\label{eq:SE_Lagrangian}
S_E&= \mathfrak{j}_\pi\left(\left\{\left(q^i,v^i,p_i,-\frac{\partial L}{\partial q^i},p_i-\frac{\partial L}{\partial v^i},v^i\right)\right\} \cap \big(\ker\, (T\text{pr}_2)\big)^\circ \right) \nonumber \\
&=\left\{\left(q^i,p_i,-\frac{\partial L}{\partial q^i},v^i\right)\st p_i=\frac{\partial L}{\partial v^i}\right\}\subset T^*T^*Q.
\end{align}

To get dynamics, we consider the generalized Dirac system determined by the pair $(D_{\omega_Q}, S_E)$. Recall that a curve $\gamma\colon I\to T^*Q$ is a solution of this implicit system if it satisfies:
\[
(\gamma(t), \dot{\gamma}(t)) \oplus \mu_{\gamma(t)} \in (D_{\omega_Q})_{\gamma(t)}\quad  \mbox{ for all }\, t\in I,
\]
where $\mu_{\gamma(t)}\in (S_E)_{\gamma(t)}$. In coordinates, $\gamma(t)=(q^i(t),p_i(t))$ and a solution satisfies
\[
\frac{d q^i}{dt}=v^i,\qquad \frac{d p_i}{dt}=\frac{\partial L}{\partial q^i},\qquad p_i-\frac{\partial L}{\partial v^i}=0.
\]
Expressed in more familiar terms,
\[
\frac{d q^i}{dt}=v^i,\qquad \frac{d}{dt}\left( \frac{\partial L}{\partial v^i}\right)=\frac{\partial L}{\partial q^i},
\]
which are the well-known Euler-Lagrange equations for $L$.

\remark We can alternatively consider the canonical Dirac structure $D_M$ on $M = TQ \oplus T^* Q$ defined by the pullback of the canonical symplectic form on $N=T^* Q$ by the canonical projection $\hbox{pr}_2\colon M\to T^*Q$ (see~\eqref{eq:D_M}). We take the usual energy on $M$, $E(q,v,p)=pv-L(q,v)$. A curve $(q(t),v(t),p(t))$ is a solution of the Dirac system $(D_M,{\rm d}E)$ if it satisfies the implicit Euler-Lagrange equations obtained above.

\paragraph{Hamiltonian mechanics.} Consider a Hamiltonian $H\colon T^*Q\to\mathbb{R}$ and the Dirac structure $D_{\omega_Q}$ on $T^*Q$. It is immediate to check that the Dirac system on $T^*Q$ with energy $E=H$ given by
\[
(\dot q,\dot p)\oplus {\rm d}E\in D_{\omega_Q},
\]
leads to the classical Hamilton equations
\[
\dot{q} = \frac{\partial H }{\partial p}(q,p), \qquad \dot{p}= -\frac{\partial H }{\partial q}(q,p).
\]

\remark More details about the Dirac structure $D_{\omega_Q}$ and its relation with the geometry of the spaces $TT^*Q$, $T^*T^*Q$ and $T^*TQ$ can be found in Appendix~\ref{ap:A}.

\subsection{Mechanics on linear almost Poisson structures}\label{subsec:almost}

Let $\tau_{\A}\colon  \A \to Q$ be a vector bundle of rank $n$ over a manifold $Q$ of dimension $m$, and let $\A^*$ be the dual vector
bundle of $\A$, with corresponding vector bundle projection $\pi_{\A^*}\colon  \A^* \to Q$. Recall that a \emph{linear almost Poisson structure} on $\A^*$ is a bracket
\[
\{ \cdot , \cdot \}_{\A^*}\colon  C^{\infty}(\A^*) \times C^{\infty}(\A^*)\to C^{\infty}(\A^*)
\]
such that:
\begin{enumerate}[i)]
\item
$\{ \cdot , \cdot \}_{\A^*}$ is skew-symmetric, that is,
\[
\{\varphi, \psi \}_{\A^*} = -\{\psi, \varphi\}_{\A^*}, \quad \mbox{
for } \varphi, \psi \in C^{\infty}(\A^*).
\]

\item
$\{\cdot , \cdot \}_{\A^*}$ satisfies the Leibniz rule, that is,
\[
\{\varphi \varphi', \psi \}_{\A^*} = \varphi \{\varphi',
\psi\}_{\A^*} + \varphi' \{\varphi, \psi \}_{\A^*}, \quad \mbox{ for
} \varphi, \varphi', \psi \in C^{\infty}(\A^*).
\]

\item
$\{\cdot , \cdot \}_{\A^*}$ is linear, which by definition means that if $\varphi$ and
$\psi$ are linear functions on $\A^*$ then
$\{\varphi, \psi \}_{\A^*}$ is also a linear function.
\end{enumerate}
If, in addition, the bracket satisfies the Jacobi identity, then  $\{\cdot, \cdot\}_{\A^*}$ is called a \emph{linear Poisson
structure} on $\A^*$. We will denote by $\Lambda_{\A^*}(df, dg)=\{f, g\}_{\A^*}$ the (almost) Poisson bivector associated to an (almost) Poisson linear structure. The associated Dirac structure will be denoted ${D}_{\A^*}\subset T\A^*\oplus T^* \A^*$.

The local description of such a bracket is as follows. Let $(q^i)$, $1\leq i \leq n$ be local coordinates on an open subset $U$ of $Q$ and $\{e^{A}\}$, $1\leq A\leq \rank \A^*$ be a local basis of sections of  $\pi_{\A^*}\colon  \A^* \to Q$. Any point $\alpha_q\in \A^*$ is locally
given by $\alpha_q=p_{A}e^{A}(q)$ and, therefore, $(q^i,p_A)$ provide coordinates on $\A^*$. With respect to this system of coordinates
on $\A^*$, the linear almost Poisson bracket has the following local expressions:
\[
\{p_{A}, p_{B}\}_{\A^*} = -{\mathcal C}_{AB}^{D}p_{D}, \qquad \{q^j, p_{A}\}_{\A^*} =
\rho^{j}_{A}, \qquad \{q^i, q^j\}_{\A^*} = 0,
\]
with $C_{AB}^{D}$ and $\rho^{j}_{A}$ real
$C^{\infty}$-functions on $U$. Consequently, the linear almost Poisson bivector associated to the linear almost Poisson structure on
$\A^*$ has the following coordinate expression:
\begin{equation*}
\Lambda_{\A^*} = \rho^{j}_{A} \displaystyle
\frac{\partial}{\partial q^j} \wedge \frac{\partial}{\partial
p_{A}} - \frac{1}{2} {\mathcal C}_{AB}^{D} p_{D}
\frac{\partial}{\partial p_{A}} \wedge
\frac{\partial}{\partial p_{B}}.
\end{equation*}

We remark that this example generalizes the case $\A=TQ$, where the Poisson bivector $\Lambda_{T^*Q}$ is the canonical one associated to the canonical symplectic form. Note that both $\omega_Q$ and $\Lambda_{T^*Q}$ define the same Dirac structure on $T^*Q$.

\paragraph{Equations of motion.} We specify the dynamics giving a Lagrangian function
$L: \A\rightarrow {\mathbb R}$ with associated energy function $E\colon M\to\R$, with $M=\A\oplus \A^*$,
\[
\begin{array}{rrcl}
E:&\A\oplus \A^*&\longrightarrow& {\mathbb R}\\
   &(v_q, \alpha_q)&\longmapsto & \langle \alpha_q, v_q\rangle-L(v_q),
\end{array}
\]
which is a Morse family over the submersion $\pi_{(M,\A^*)}\colon M \rightarrow \A^*$. This Morse family generates the immersed Lagrangian submanifold $S_E$ of the symplectic manifold $(T^*\A^*, \omega_{\A^*})$. A curve $\gamma\colon  I\subset \mathbb{R}\to \A^*$ is a solution of generalized Dirac system $(D_{\A^*}, S_E)$ if it satisfies
\[
(\gamma(t), \dot{\gamma}(t)) \oplus \mu_{\gamma(t)} \in (D_{\A^*})_{\gamma(t)} \quad \forall\, t\in I.
\]
Taking coordinates $(q^i, v^A)$ on $\A$ induced by the dual basis $\{e_{A}\}$ of $\{e^A\}$, we
have coordinates $(q^i, v^A, p_A)$ on $\A\oplus \A^*$. Then $\hbox{pr}_2(q^i, v^A, p_A)=(q^i, p_A)$,  $E(q^i, v^A, p_A)=p_Av^A-L(q^i, v^A)$, and by a computation completely analogous to the one used to obtain~\eqref{eq:SE_Lagrangian}, we find:
\[
S_E=\left\{(q^i, p_A, \alpha_i, \beta^A) \in T^* \A^*
\st
\alpha_i =-\frac{\partial L}{\partial q^i},\, \beta^A = v^A,\, p_A - \frac{\partial L}{\partial v^A} = 0\right\}.
\]
Note that, in general, $S_E$ will not be the graph of the differential of a real function on $\A^*$ (but it is an immersed Lagrangian submanifold, as shown in Section~\ref{sec:Morse}). If the Lagrangian is regular, then $S_E$ can be obtained as a differential.

The expression of  $D_{\Lambda_{\A^*}}\subset T\A^*\oplus T^*\A^*$ is obtained as the graph of the bivector $\Lambda_{\A^*}$:
\[
D_{\Lambda_{\A^*}}=\{(q^i,p_A,\dot q^i,\dot p_A,\alpha_i,\beta^A)\st \dot q^i=\rho^i_A\beta^A,\, \dot p_A=-\rho^i_A\alpha_i-C^D_{AB}p_D\beta^B\}.
\]
The equations of motion will follow from the generalized Dirac system $(D_{\Lambda_{\A^*}},S_E)$. In coordinates, a curve solves the Dirac system $(D_{\A^*},S_E)$ if, and only if,
\[
\dot q^i= \rho^i_A v^A,\qquad \dot p_A=\rho^j_A\frac{\partial L}{\partial q^j}-C^D_{AB}p_Dv^B, \qquad p_A - \frac{\partial L}{\partial v^A} = 0.
\]
For the Hamiltonian description, the energy $E$ is given by the Hamiltonian $H(q,p)$, and the Dirac system $(D_{\A^*},{\rm d}E)$ leads to:
\[
\dot q^i= \rho^i_A v^A,\qquad \dot p_A=-\rho^j_A\frac{\partial H}{\partial q^j}-C^D_{AB}p_D\frac{\partial H}{\partial p_B},
\]
which agree with those in the literature, see~\cite{LeMaMa_LinearAlmostPoisson}.

\remark More details on the geometry of $D_{\Lambda_{\A^*}}$ can be found in Appendix~\ref{ap:A}.

\paragraph{Euler-Poincar\'e equations.} A particular case of the previous construction is obtained when $\A^*$ is the dual of a Lie algebra (as a vector bundle over a point). Let $G$ be a Lie group, ${\mathfrak g}$ its Lie algebra and  ${\mathfrak g}^*$ its dual. On ${\mathfrak g}^*$ we have  the $\pm$ Lie-Poisson bracket:
\[
\left\{f, g\right\}_{\pm}(\mu)=\pm \left\langle \mu, \left[ \frac{\delta f}{\delta \mu}, \frac{\delta g}{\delta \mu}\right]\right\rangle,
\]
where $\mu\in {\mathfrak g}^*$ and $\frac{\delta f}{\delta \mu}: {\mathfrak g}^*\rightarrow {\mathfrak g}$ stands for the functional derivative of $f$ and where $[\cdot,\cdot]$ is the Lie algebra bracket on ${\mathfrak g}$. It is well known that both brackets are induced by reduction of the standard Lie bracket on $T^*G$ by right or left-reduction. Let us denote by ${\mathfrak g}^{\pm}$ the Lie algebra ${\mathfrak g}$ endowed with the  $(\pm)$-Lie-Poisson bracket. The map $\sharp_{{\mathfrak g}^{\pm}}$ is simply
\[
\begin{array}{rcl}
\sharp_{{\mathfrak g}^{\pm}}: T^*{\mathfrak g}^*\equiv {\mathfrak g}^*\times {\mathfrak g}&\longrightarrow& T{\mathfrak g}^* \equiv {\mathfrak g}^*\times {\mathfrak g}^*\\
(\mu, \xi)&\longmapsto& (\mu, \mp ad^*_{\xi}\mu),
\end{array}
\]
and we have a Dirac structure defined on ${\mathfrak g}^*$, ${D}_{{\mathfrak g}^{\pm}}\subset T{\mathfrak g}^*\oplus T^* {\mathfrak g}^*\equiv {\mathfrak g}^*\oplus {\mathfrak g}^*\oplus {\mathfrak g}$. Given a Lagrangian $L: {\mathfrak g}\longrightarrow {\mathbb R}$, frequently defined as a reduced Lagrangian from $TG$, we have the Morse family
\[
\begin{array}{rcl}
E: {\mathfrak g}\oplus {\mathfrak g}^*&\longrightarrow& {\mathbb R}\\
(\xi, \mu)&\longmapsto& \langle \mu, \xi\rangle -L(\xi),
\end{array}
\]
generating the Lagrangian submanifold $S_E$ of $T^* {\mathfrak g}^*\equiv {\mathfrak g}^*\times {\mathfrak g}$ given by
\[
S_E=\left\{(\mu, \xi)\in {\mathfrak g}^*\times {\mathfrak g}\st \mu=  \frac{\delta L}{\delta \xi}\right\}.
\]
The solutions of the generalized Dirac system $({D}_{{\mathfrak g}^{\pm}}, S_E)$ are curves $\mu: I\subset {\mathbb R}\longrightarrow {\mathfrak g}^*$ such that
\[
(\mu(t), \dot{\mu}(t), \xi(t)) \in ({D}_{{\mathfrak g}^{\pm}})_{\mu(t)}\quad  \mbox{ for all }\, t\in I,
\]
where $(\mu(t),\xi(t))\in S_E$. Therefore, the equations are:
\[
\frac{d\mu}{dt}= \mp ad^*_{\xi}\mu, \quad  \mu =  \frac{\delta L}{\delta \xi},
\]
which correspond to the Euler-Poincar\'e equations
\[
\frac{d}{dt}\left(\frac{\delta L}{\delta \xi}\right)= \mp ad^*_{\xi}\frac{\delta L}{\delta \xi}\; .
\]
We refer the reader to~\cite{Holm_Book,MR_book} for more details.

\paragraph{Euler-Poincar\'e equations with advected parameters.} Another important class of examples comes from actions of Lie algebras on manifolds.  Given a  homomorphism  $\Phi $ from the Lie algebra ${\mathfrak g}$ to the Lie algebra of vector fields on $Q$, ${\mathfrak X}(Q)$, we can induce a linear Poisson bracket on the trivial bundle ${\A}^*=Q \times {\mathfrak g}^*\rightarrow Q$ as follows:
\begin{eqnarray*}
\{f, g\}_{{\A}^*}&=&-\left\langle \mu, \left[ \frac{\delta f}{\delta \mu}, \frac{\delta g}{\delta \mu}\right]\right\rangle+dg_{q}\left(\Phi \left(\frac{\delta f}{\delta \mu}\right)\right)
 -df_{q}\left(\Phi \left(\frac{\delta g}{\delta \mu}\right)\right).
\end{eqnarray*}
Here $d_qf$ stands for the differential of $f$ with respect to $q\in Q$, and the evaluation point $(q, \mu)$ has been suppressed. It is not difficult to check that
\[
\begin{array}{rcl}
\sharp_{{\A}^*}: T^*(Q\times {\mathfrak g}^*)\equiv T^* Q\times {\mathfrak g}^*\times {\mathfrak g}&\longrightarrow&T(Q\times {\mathfrak g}^*)\equiv TQ\times {\mathfrak g}^*\times {\mathfrak g}^*\\
(\alpha_q, \mu, \xi)&\longmapsto& (  -\Phi(\xi)_q,\mu, ad^*_{\xi}\mu+{\mathbf J} (\alpha_q)),
\end{array}
\]
where ${\mathbf J}: T^*Q\rightarrow {\mathfrak g}^*$ is the associated cotangent momentum map $\langle {\mathbf J}(\alpha_q), \xi\rangle=\langle \alpha_q, \Phi(\xi)_q\rangle$.

For a Lagrangian $L: {\A}=Q \times {\mathfrak g}\rightarrow {\mathbb R}$ we define the Morse family $E: Q\times {\mathfrak g}\times {\mathfrak g}^*\longrightarrow {\mathbb R}$ given by
$E(q, \xi, \mu)=\langle \mu, \xi\rangle - L(q,\xi)$ which generates
 the following Lagrangian submanifold $S_E$ of $T^* (Q\times {\mathfrak g}^*)=T^* Q\times {\mathfrak g}^*\times {\mathfrak g}$
\[
S_E=\left\{(\alpha, \mu, \xi)\in T^* Q\times {\mathfrak g}^*\times {\mathfrak g}\st  \alpha=\frac{ \delta L}{\delta q},\,  \mu=  \frac{\delta L}{\delta \xi}\right\}.
\]
The equations of motion are derived now using the generalized Dirac system $(D_{{\A}^*}, S_E)$ and the corresponding  solutions  are curves $(q,\mu): I\subset {\mathbb R}\longrightarrow Q\times {\mathfrak g}^*$
such that
\[
(q(t), \mu(t), \dot{q}(t), \dot{\mu}(t), \xi(t)) \in (D_{{\A}^*})_{\{q(t), \mu(t)\}} \quad \forall\,  t\in I,
\]
where $(q(t), \mu(t),\xi(t))\in S_E$. Therefore, the equations are:
\[
\frac{dq}{dt}=-\Phi(\xi(q)),\quad \frac{d\mu}{dt}= ad^*_{\xi}\mu-{\mathbf J} (\alpha), \quad \mu =  \frac{\delta L}{\delta \xi},\quad \alpha=\frac{ \delta L}{\delta q},
\]
which corresponds to the following equations
\[
\frac{d}{dt}\left(\frac{\delta L}{\delta \xi}\right)= ad^*_{\xi}\frac{\delta L}{\delta \xi}+{\mathbf J} \left(\frac{ \delta L}{\delta q}
\right), \quad \frac{dq}{dt}=-\Phi(\xi(q)).
\]

As a particular case, suppose that $G$ is a Lie group acting by left representation on a vector space $V$, and denote by $v\mapsto gv$ the left representation of $g\in G$ on $v\in V$.
Then,  $G$ also acts on the left on its dual space $V^*$. For each $v\in V$, denote by $\rho_v: {\mathfrak g}\rightarrow V$ the  linear map given by
\[
\rho_v(\xi)=\left.\frac{d}{dt}\right |_{t=0} \hbox{exp}(\xi t)v ,
\]
and denote by $\rho_v^*: V^*\rightarrow {\mathfrak g}^*$ the map
\[
\langle \rho_v^*(a), \xi\rangle=\langle a,  \rho_v(\xi)\rangle,\quad a\in V^*\; , \ \xi\in {\mathfrak g}.
\]
 We will use the common notation
\[
\rho_v^*a=v \diamond a\in {\mathfrak g}^*.
\]

Particularizing the previous discussion to this case, we obtain that ${\A}=V^*\times {\mathfrak g}\rightarrow V^*$ where now the homomorphism  $\Phi: {\mathfrak g}\rightarrow \hbox{End } (V^*)$,
is given by
\[
\langle \Phi(\xi)a, v\rangle=\langle v \diamond a, \xi\rangle.
\]
In this case ${\A}^* =V^*\times {\mathfrak g}^*\rightarrow V^*$ and the linear Poisson structure is given by
\[
\begin{array}{rcl}
\sharp_{{\A}^*}:  V^*\times V\times {\mathfrak g}^*\times {\mathfrak g}&\longrightarrow&V^*\times V^*\times {\mathfrak g}^*\times {\mathfrak g}^*\\
(a, v, \mu, \xi)&\longmapsto& (a, -\Phi(\xi)a,  v \diamond a,\mu, ad^*_{\xi}\mu+v \diamond a).
\end{array}
\]
Given a Lagrangian $L: V^*\times {\mathfrak g}\rightarrow {\mathbb R}$ we obtain the Euler-Poincar\'e equations with advected parameters
\[
\frac{d}{dt}\left(\frac{\delta L}{\delta \xi}\right)=  ad^*_{\xi}\frac{\delta L}{\delta \xi}+\left(\frac{ \delta L}{\delta a}\diamond a
\right),\quad \frac{da}{dt}=-\Phi(\xi)a.
\]
We refer the reader to~\cite{HoMaRa_Semidirect} for many interesting applications of these equations. See also~\cite{CeMa} and~\cite{CeIbMa} for a discussion of variational principles in this context.

\paragraph{Nonholonomic mechanics.}
Now we will show that the formalism of Dirac systems covers some interesting cases of nonholonomic mechanics (see~\cite{Bloch} and references therein).  This particular type of constrained mechanical systems have a considerable practical  interest  since
nonholonomic constraints are present in a great variety of engineering and robotic tools describing the dynamics of wheeled vehicles, manipulation gadgets and locomotion, etc.
For simplicity and since  most examples are in that category, we will only consider systems subjected to linear constraints. In this case, the dynamics is described by a $C^{\infty}$-distribution  $\A$ on the configuration space $Q$, that is, by a linear vector subbundle
$\tau_{\A}\colon \A\rightarrow {\mathbb R}$ of $TQ$ with canonical inclusion $i_{\A}\colon \A\rightarrow TQ$.
We say that $\A$ is holonomic if $\A$  is integrable or involutive and nonholonomic otherwise, that is, a regular linear velocity
constraint $\A$ is nonholonomic if it is not holonomic.
Therefore, we will say that a curve $\gamma\colon I\subseteq {\mathbb R}\rightarrow Q$ satisfies the constraints given by $\A$  if
\[
\dot{\gamma}(t)=\frac{d\gamma}{dt}(t)\in \A_{\gamma(t)}=\tau^{-1}_{\A}(\gamma(t))\; .
\]

Mathematically a mechanical nonholonomic system is given by the following data:
\begin{enumerate}[1)]
\item  A Lagrangian $L\colon TQ\rightarrow {\mathbb R}$ of the form:
\[
L(v_q)=\frac{1}{2}g(v_q, v_q) - V(q), \quad v_q\in T_qQ.
\]
Here $g$ denotes a Riemannian metric on the configuration space $Q$ and $V\colon Q\rightarrow {\mathbb R}$ is a potential function. Locally, the metric is determined by the non-degenerate symmetric matrix $(g_{ij})_{1\leq i,j\leq \dim Q}$ such that $g=g_{ij}(q)\, dq^i\otimes d q^j$. Therefore
\[
L(q^i, \dot{q}^i)=\frac{1}{2}g_{ij}\dot{q}^i\dot{q}^j-V(q).
\]
\item A vector subbundle $\tau_{\A}\colon \A\rightarrow Q$.
\end{enumerate}

\noindent Given this two data, the equations of motion  in nonholonomic mechanics are completely determined by the Lagrange-d'Alembert
principle. This principle states that a curve $\gamma\colon I\subseteq  {\mathbb R}\rightarrow Q$  is an admissible motion of the system if
\[
\delta \int_0^T L(\gamma(t), \dot{\gamma}(t))\; dt=0
\]
for variations which satisfy $\delta \gamma\in \A_{\gamma(t)}$. Locally, if $\gamma(t)=q^i(t)$ then, from the Lagrange-d'Alembert principle, we arrive at the well-known nonholonomic equations
\begin{eqnarray*}
\frac{d}{dt}\left(
\frac{\partial L}{\partial \dot{q}^i}\right)-\frac{\partial L}{\partial q^i}&=&\lambda_{\alpha} \mu^{\alpha}_i,\\
\mu^{\alpha}_i\dot{q}^i&=&0,
\end{eqnarray*}
where $\A^0=\hbox{span}\{\mu^{\alpha}_i\, d{q}^i, 1\leq \alpha\leq  \hbox{rank}\, \A\}$.

To present the previous nonholonomic equations as a generalized Dirac system we proceed as follows (see~\cite{GrLeMaMa, LeMaMa_LinearAlmostPoisson}). Using the Riemannian metric $g$ we have the orthogonal decomposition
$TQ=\A\oplus \A^{\perp}$ and the corresponding
 orthogonal projectors
\begin{align*}
{\mathcal P}\colon &TQ\to  \A,\\
{\mathcal Q}\colon &TQ\to
\A^{\perp}.
\end{align*}

Now define the following linear almost Poisson bracket  $\{\cdot , \cdot \}_{{\A}^\ast}$ on ${\A}^*$ by:
\[
\{f, g\}_{{\A}^*}=\{f\circ i_{\A}^*, g\circ i_{\A}^*\}_{T^*Q}\circ {\mathcal P}^*,
\]
where $i_{\A}:\A\to TQ$ denotes the inclusion, $\{\cdot , \cdot \}_{T^*Q}$ is the standard Poisson bracket of the cotangent bundle and $f, g\in C^{\infty}(\A^*)$. Denote by $\Lambda_{\A^*}(df, dg)=\{f, g\}_{{\A}^\ast}$, and by $D_{\Lambda_{\A^*}}$  the associated Dirac structure on this almost Poisson manifold
(see  \cite{LeMaMa_LinearAlmostPoisson}, and also~\cite{VaMa, CaLeMa1}). Defining $E: \A\oplus \A^*\rightarrow {\mathbb R}$ by
\[
E(v_q, \alpha_q)=\langle \alpha_q, v_q\rangle-L(v_q).
\]
the solutions of the nonholonomic equations are expressed as the solutions of the generalized Dirac system determined by $(D_{\Lambda_{\A^*}}, {\rm d}E)$.

Locally,  consider a basis of vector fields $\{X_a, X_{\alpha}\}$ on $Q$, $1\leq a \leq m$ and   $m+1\leq  \alpha \leq n$,
such that
\[
\A_q=\hbox{span }\{ X_a(q)\} \   \hbox{ and }\ \A^\perp_q=\hbox{span }\{ X_{\alpha}(q)\}.
\]
Using this decomposition it is easy to induce coordinates $(x^i, y^a)$ in $\A$ such that for any vector $v_q\in T_qQ$
\[
v_q=y^aX_a(q)+y^{\alpha} X_{\alpha}(q).
\]
These coordinates are known as quasi-velocities with respect to the frame $\{X_a,X_\alpha\}$, and they turn out to be very useful in the context of nonholonomic mechanics (we refer the reader to~\cite{CraMest_Nonholonomic} for more details; see also~\cite{BloMarZenk_Quasi}). Observe that in these coordinates, if $v_q\in \A_q$, then $y^{\alpha}=0$ represents the nonholonomic constraints.  Consequently, $\A$ is described by the coordinates $(q^i, y^a)$. In these adapted coordinates the nonholonomic equations are rewritten as:
\begin{eqnarray*}
&&\frac{d}{dt}\left(\frac{\partial l}{\partial y^a}\right)+{\mathcal C}^c_{ab}y^b\frac{\partial l}{\partial y^c}-X^i_a\frac{\partial l}{\partial q^i}=0,\\
&&\dot{q}^i=X^i_a(q)y^a,
\end{eqnarray*}
where $X_a=X^i_a(q)\frac{\partial}{\partial q^i}$, ${\mathcal P}[X_a, X_b]={\mathcal C}_{ab}^c X_c$  and $l\colon {\mathcal A}\rightarrow {\mathbb R}$ is the restriction of $L$ to $\A$, that is, $L_{|{\A}}=l$. The almost Poisson bracket satisfies:
\[
\{p_a,p_b\}=-C^c_{ab}p_c,\qquad \{q^i,p_a\}=X^i_a,\qquad \{q^i,q^j\}=0,
\]
where the $\{p_a\}$ are implicitely defined by the relation $p_a=\frac{
\partial l}{\partial y^a}$.

\subsection{Constrained variational calculus on linear almost Poisson manifolds}

Let $L: M\rightarrow \mathbb{R}$ be a constrained Lagrangian where  $M \subset TQ$ is a constraint submanifold. A constrained variational problem \cite{Arnold88}, \cite{2002CLMM}, \cite{2008IMMS} consists on finding critical points of an action functional
\[
 \int_{t_0}^{t_1} L(q(t), \dot{q}(t))\,{\rm d}t
\]
on the family of curves satisfying some  fixed endpoints condition as, for instance, $q(t_0) = q_0$, $q(t_1) = q_1$ and, besides, satisfying the constraints, that is, $(q(t), \dot{q}(t))\in M_{q(t)}$, for all $t \in (t_0, t_1)$. The submanifold $M$ is  $(2n-m)-$dimensional and is locally determined by the vanishing of constraint functions, $\Phi^a=0$, $1\leq a\leq m$, where $\Phi^a:TQ\rightarrow {\mathbb R}$.

We are implicitly  assuming that the solution curves $q(t)$ admit nontrivial variations in the space of curves satisfying the constraints; that is, we are dealing with normal solutions,  in opposition to the \emph{abnormal} ones, which are pathological curves that do not admit nontrivial variations.

In the case of normal solutions it is possible to characterize the solutions by  using the standard  procedure of Lagrange multipliers. The usual way to present the equations of motion of vakonomic mechanics is the following:

\begin{equation}\label{eq:asdo}
\left\{
\begin{array}{l}
\displaystyle{ \frac{d}{d t}\left( \frac{\partial \tilde{L}}{\partial \dot{q}^i} \right) -
\frac{\partial \tilde{L}}{\partial q^i}  = \dot{\lambda}_{a} \frac{\partial \Phi^{a}}{\partial \dot q^i}
 + \lambda_{a} \left[ \frac{d}{d t}\left( \frac{\partial \Phi^{a}}{\partial \dot{q}^i} \right) -
\frac{\partial \Phi^{a}}{\partial q^i}\right]}\,,\\
\\
\Phi^{a}(q, \dot{q}) = 0, \ 1 \leq a \leq m \,,
\end{array}
\right.
\end{equation}
where $\lambda_a$ are  Lagrange multipliers to be determined and $\tilde{L}\colon TQ\rightarrow {\mathbb R}$ is an arbitrary extension of $L$ to $TQ$. The equations~\eqref{eq:asdo} can be seen as the Euler-Lagrange equations for the extended Lagrangian $\mathcal{L}=\tilde{L}+\lambda_a\Phi^a$. Note that if we consider the extended Lagrangian $\lambda_0\tilde L+\lambda_a\Phi^a$, with $\lambda_0=0$ or $1$, then we recover all the solutions, both the normal and the abnormal ones \cite{Arnold88}.

We will see how our scheme is easily adapted to the case of constrained variational problems. Assume, for sake of simplicity,  that the restriction $(\tau_Q)_{|M}\colon M\rightarrow Q$ is a surjective submersion. In this case,  we can choose  coordinates $(q^i, \dot{q}^A)$ on $M$ and the constraints are rewritten $\Phi^a(q^i, \dot{q}^i)=\varphi^a (q^i, \dot{q}^A)-\dot{q}^a$. In other words, $M$ admits the local description:
\[
M= \{(q^i,\dot q^A, \dot q^a) | \;  \dot q^a=\varphi^a(q^i, \dot q^A)\} \, .
\]

Given the Lagrangian $L\colon M\rightarrow {\mathbb R}$, define the function $E: M\times_Q T^*Q\rightarrow {\mathbb R}$ by
\[
E(v_q, \alpha_q)=\langle \alpha_q, v_q\rangle -L(v_q)\;,
\]
where $v_q\in M_q$ and $\alpha_q\in T_q^* Q$.
In adapted coordinates
\[
E(q^i, \dot{q}^A, p_i)=\langle p_A, \dot{q}^A\rangle +\langle p_a,  \varphi^a (q^i, \dot{q}^A)\rangle -L(q^i, \dot{q}^A)\; .
\]
It is a simple computation to show that $E$ is a Morse family for the projection $\hbox{pr}_2: M\times_Q T^*Q\rightarrow T^*Q$, and it generates the following immersed Lagrangian submanifold $S_E$ of $T^*T^*Q$:
\begin{align*}
S_E=&\left\{(q^i, p_i, \alpha_i, \beta^i)\in T^*T^*Q\st \alpha_i=\frac{\partial E}{\partial q^i},\, \beta^i=\frac{\partial E}{\partial p_i},\,
\frac{\partial E}{\partial v^A}=0\right\}\\
=&\left\{(q^i, p_i, \alpha_i, \beta^i)\in T^*T^*Q\st \alpha_i=-\frac{\partial L}{\partial q^i}+p_a\frac{\partial \varphi^a}{\partial q^i},\, \beta^A=v^A,\, \beta^a=\varphi^a(q^i, \dot{q}^A)\right.,\\
&\;\;\left.p_A+p_a\frac{\partial \varphi^a}{\partial v^A}-\frac{\partial L}{\partial v^A}=0\right\}.
\end{align*}

If we consider the Dirac structure $D_{\omega_Q}$ on $T^*Q$, then a curve $\gamma(t)=(q^i(t), p_i(t))$ is a solution of the Dirac system $(D_{\omega_Q}, S_E)$ if
\[
\frac{d q^A}{dt}=v^A,\quad \frac{d q^a}{dt}=\varphi^a(q^i, \dot{q}^A),\quad \frac{d p_i}{dt}=\frac{\partial L}{\partial q^i},\quad
p_A=\frac{\partial L}{\partial v^A}-p_a\frac{\partial \varphi^a}{\partial v^A}.
\]

In other terms,
\begin{eqnarray*}
\frac{d q^i}{dt}&=&v^i,\\
\frac{d q^a}{dt}&=&\varphi^a(q^i, \dot{q}^A),\\
\frac{d p_a}{dt}&=&\frac{\partial L}{\partial q^a}-p_b\frac{\partial \varphi^b}{\partial q^a},\\
\frac{d }{dt}\left(\frac{\partial L}{\partial v^A}-p_a\frac{\partial \varphi^a}{\partial v^A}\right)&=&\frac{\partial L}{\partial q^i}-p_a\frac{\partial \varphi^a}{\partial q^A}.
\end{eqnarray*}
These equations are equivalent to the equations~\eqref{eq:asdo}, where $\Psi^a(q^i, \dot{q}^i)=\varphi^a(q^i, \dot{q}^A)-\dot{q}^a$ and $\lambda_a=\frac{\partial L}{\partial q^a}-p_a$.

\paragraph{The case of a general vector bundle.} The same procedure works in the case of reduced constrained systems where we have a Dirac structure on a linear almost Poisson manifold ${\A}^*$, denoted ${D}_{\Lambda_{{\A}^*}}$, as discussed earlier in Section~\ref{subsec:almost}. Besides the vector bundle $\tau_{\A}: {\A}\rightarrow  Q$, we assume that we have a fiber bundle  $\tau_M: M\rightarrow Q$ with $M\subset {\A}$ (which, in general, is not a vector subbundle) and a Lagrangian $L: M\rightarrow {\mathbb R}$ .

With the notation used above, let $\hbox{pr}_2$ denote the projection
 \[
\hbox{pr}_2: M\times_Q {\A}^*\longrightarrow {\A}^*,
 \]
and take as a Morse family
 \[
 \begin{array}{rrcl}
 E:& M\times_Q {\A}^*&\longrightarrow& {\mathbb R},\\
     &(v_q, \alpha_q)&\longmapsto& \langle \alpha_q, v_q\rangle -L(v_q),
     \end{array}
     \]
where $v_q\in M_q$ and $\alpha_q\in \A^*_q$. The equations corresponding to the generalized Dirac system determined by the pair $({D}_{\Lambda_{{\A}^*}}, S_E)$ are:
\[
(\gamma(t), \dot{\gamma}(t)) \oplus \mu_{\gamma(t)} \in ({D}_{\Lambda_{{\A}^*}})_{\gamma(t)} \quad   \forall\, t\in I,
\]
where $\mu_{\gamma(t)}\in (S_E)_{\gamma(t)}$. In order to write locally the equations of motion, we will choose local fiber bundle coordinates
$(q^i,y^A, y^a)$ of ${\A}$ such that
\begin{equation*}
M= \{(q^i,y^A, y^a) | \;  y^a=\varphi^a(q^i, y^A)\} \, ,
\end{equation*}
and we obtain
\begin{align*}
E(q^i,y^A, y^a, p_A, p_a)
&= p_A y^A+p_a\varphi^a(q^i,y^A)-L(q^i, y^A).
\end{align*}
Observe now that
\begin{align*}
S_E=&\left\{(q^i, p_A, p_a, \alpha_i, \beta^A, \beta^a)\in T^*{\A}^*\st \alpha_i=\frac{\partial E}{\partial q^i},\, \beta^A=\frac{\partial E}{\partial p_A},\,  \beta^a=\frac{\partial E}{\partial p_a},\,
\frac{\partial E}{\partial v^A}=0\right\}\\
=&\left\{(q^i, p_A, p_a, \alpha_i, \beta^A, \beta^a)\in T^*{\A}^*\; |\; \alpha_i=-\frac{\partial L}{\partial q^i}+p_a\frac{\partial \varphi^a}{\partial q^i},\, \beta^A=y^A,\, \beta^a=\varphi^a(q^i, v^A)\right.,\\
&\;\;\left.p_A+p_a\frac{\partial \varphi^a}{\partial y^A}-\frac{\partial L}{\partial y^A}=0\right\}.
\end{align*}
and the solutions $\gamma(t)=(q^i(t), p_A(t), p_a(t))$ of the Dirac system $({D}_{\Lambda_{{\A}^*}}, S_E)$ verify the following systems of equations:
\begin{eqnarray*}
 \dot{q}^i&=&\rho^i_A y^A +\rho^i_a \varphi^a(q^i,y^A),\\
 \dot{p}_{a}&=&\rho^j_{a}\frac{\partial L}{\partial q^j}-\rho^j_a p_b\frac{\partial \varphi^b}{\partial q^j}
-{C}_{aB}^{D} p_{D}y^{B}-{C}_{aB}^{c} p_{c}y^{B}\\
&&-{C}_{ab}^{D} p_{D}\varphi^b(q^i, y^A)-{C}_{ab}^{c} p_{c}\varphi^b(q^i, y^A),\\
\dot{p}_{A}&=&\rho^j_{A}\frac{\partial L}{\partial q^j}-\rho^j_A p_b\frac{\partial \varphi^b}{\partial q^j}
-{C}_{AB}^{D} p_{D}y^{B}-{C}_{AB}^{c} p_{c}y^{B}\\
&&-{C}_{Ab}^{D} p_{D}\varphi^b(q^i, y^B)-{C}_{Ab}^{c} p_{c}\varphi^b(q^i, y^B),\\
 p_{A}&=&\frac{\partial L}{\partial y^A}-p_a\frac{\partial \varphi^a}{\partial y^A}.
\end{eqnarray*}
These equations can be found in~\cite{2008IMMS}.

\paragraph{Higher-order mechanics.} Our geometric approach also recovers the higher-order mechanics whose Lagrangian function is given by: $L\colon T^{(k)}Q \rightarrow \mathbb{R}$ and we have the Dirac structure $D_{\omega_{T^*T^{(k-1)}Q}}$ on $T^*T^{(k-1)}Q$ defined by the natural symplectic structure on the cotangent bundle $T^*T^{(k-1)}Q$ over $T^{(k-1)}Q$. It is important to note that the manifold $T^{(k)}Q$ can be embedded into $TT^{(k-1)}Q$ fitting into the following commutative diagram:
\begin{figure}[H]
\centering
\includegraphics{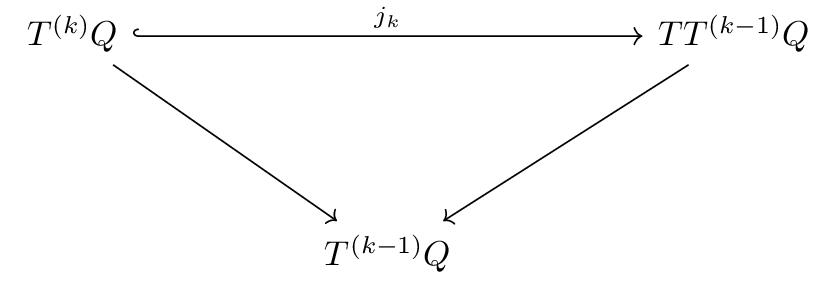}
\end{figure}
Hence, we can define the Morse family as follows
 \[
 \begin{array}{rrcl}
 E:& T^{(k)}Q\times_{T^{(k-1)}Q }T^*T^{(k-1)}Q&\longrightarrow& {\mathbb R},\\
     &(v, \alpha)&\longmapsto& \langle \alpha, j_k(v)\rangle -L(v),
     \end{array}
     \]
where $v\in T^{(k)}Q$ and $\alpha\in T^*T^{(k-1)}Q$. The reader can obtain as an exercise the equations of motion corresponding to the generalized Dirac system determined by $(D_{\omega_{T^*T^{(k-1)}Q}},S_E)$ and check that they are precisely the equations for the higher-order mechanics in~\cite{BookLeon}. We point out that in this case $S_E$ is a submanifold of $T^*T^*T^{(k-1)}Q$.

\subsection{Optimal control theory}

We will adopt the notation of Section~\ref{subsec:almost}. We consider the optimal control problem of an autonomous system with fixed initial and final boundary conditions $q_0$ and $q_T$, where $[0, T]\subset {\mathbb R}$ is a fixed interval. The set of admissible controls are piecewise-continuous functions of time taking values on a set $U\subset {\mathbb R}^m$.
The state or control equations  and associated boundary conditions have the form
\[
\dot{q}=F(q, u)\, ,\quad q(0)=q_0, \quad   q(T)=q_T,
\]
and the cost functional is:
\[
\int^T_0 L(q(t), u(t))\, dt.
\]
The optimal control problem consists of finding the minimum value of the cost functional over the control set and to determine  the solution of the state equations  for this optimal control. The standard way to single out the solution candidates is via Pontryagin's maximum principle \cite{Pontryagin_book}, but we will show here how to characterize them using our framework.

Geometrically, an optimal control problem is determined by a control bundle $\tau_C: C\rightarrow Q$ (typically, $C=Q\times U$), a fibered mapping $F: C\rightarrow TQ$ such that $\tau_C=\tau_Q\circ F$ and a cost function $L\colon  C\to {\mathbb R}$. Consider the bundle $C\mbox{$\;$}_{\tau_C} \kern-3pt\times_{\pi_{Q}}  T^*Q$
and the projection
 \[
\hbox{pr}_2: C\mbox{$\;$}_{\tau_C} \kern-3pt\times_{\pi_{Q}}  T^*Q\longrightarrow T^*Q.
   \]
The function $E: C\mbox{$\;$}_{\tau_C} \kern-3pt\times_{\pi_{Q}}  T^*Q\longrightarrow {\mathbb R}$ defined by
\[
E(u_q, \alpha_q)=\langle \alpha_q, F(u_q)\rangle-L(u_q), \quad u_q\in C_q,\qquad \alpha_q\in T_q^*Q\, ,
\]
is not in general a Morse family over ${\rm pr}_2$. Locally, $E(q^i,u^a, p_i)=p_i F^i(q^j, u^b)-L(q^j,u^b)$ (Pontryagin's Hamiltonian).
The matrix
\begin{equation*}
\begin{pmatrix}  p_j \dfrac{\partial^2 F^j}{\partial q^i\partial u^a}-\dfrac{\partial^2L}{\partial q^i \partial u^a}, &  p_j \dfrac{\partial^2 F^j}{\partial u^a\partial u^b}-\dfrac{\partial^2L}{\partial u^a \partial u^b}, & \dfrac{\partial F^j}{\partial u^a}\end{pmatrix}
\end{equation*} does not necessarily have maximum rank and thus it is not, in general, a Morse family.
It may happen that $S_E$  is not a  immersed submanifold of $T^*T^*Q$.
In any case, we can consider the system determined  by $({D}_{\omega_Q},S _E)$ and the solutions are now given by:
\[
(\gamma(t), \dot{\gamma}(t)) \oplus \mu_{\gamma(t)} \in ({D}_{\omega_Q})_{\gamma(t)} \quad \forall\, t\in I,
\]
where $\mu_{\gamma(t)}\in (S_E)_{\gamma(t)}$.
Locally,
\begin{align*}
S_E=&\left\{(q^i, p_i, \alpha_i,\beta^i)\in T^* T^* Q \st \exists  u\in U \; \mbox{such that }\;
\alpha_i=p_j \dfrac{\partial F^j}{\partial q^i}-\dfrac{\partial L}{\partial q^i},\right. \\
&\;\;\left.\beta^i=F^i(q, u),\, p_i \dfrac{\partial F^i}{\partial u^a}-\dfrac{\partial L}{\partial u^a}=0\right\}.
\end{align*}
Following the same procedure as in the previous sections  we obtain the equations of motion for the system $(D_{\omega_Q},S_E)$
\begin{eqnarray*}
\dot{q}^j&=& F^j(q, u) =\frac{\partial E}{\partial p_j}\, ,\\
\dot{p}_i &=&\left(\dfrac{\partial L}{\partial q^i}-p_j\dfrac{\partial F^j}{\partial q^i} \right)=-\frac{\partial E}{\partial q^i}, \\
0&=&p_i \dfrac{\partial F^i}{\partial u^a}-\dfrac{\partial L}{\partial u^a}=\frac{\partial E}{\partial u^a}\,,
\end{eqnarray*}
which are the typical Pontryagin's equations for the Hamiltonian function $E: C\mbox{$\;$}_{\tau_C} \kern-3pt\times_{\pi_{Q}}  T^*Q\longrightarrow {\mathbb R}$.

\paragraph{The case of general bundles.} An interesting generalization of the previous optimal control problem consists of replacing $TQ$ by a vector bundle $\A$ over $Q$. That is the following data are given:
a control bundle $\tau_C: C\rightarrow Q$ and the fibered mapping $F: C\rightarrow \A$ (such that $\tau_C=\tau_A\circ F$) ,
a cost function $L\colon  C\to {\mathbb R}$ and  the   Dirac structure  $D_{\Lambda_{{\A}^*}}$ defined in  Section~\ref{subsec:almost}. Consider the bundle
$C\mbox{$\;$}_{\tau_C} \kern-3pt\times_{\pi_{{\A}^*}}  {\A}^*$
and the projection
\[
\hbox{pr}_2: C\mbox{$\;$}_{\tau_C} \kern-3pt\times_{\pi_{{\A}^*}}  {\A}^*\longrightarrow {\A}^*.
\]
The function $E: C\mbox{$\;$}_{\tau_C} \kern-3pt\times_{\pi_{{\A}^*}}  {\A}^*\longrightarrow {\mathbb R}$ is defined by
\[
E(u_q, \alpha_q)=\langle \alpha_q, F(u_q)\rangle-L(u_q)\, , \quad u_q\in C_q\, ,\quad \alpha_q\in {\A}^*_q.
\]
Locally
\begin{align*}
S_E=&\left\{(q^i, p_A, \alpha_i,\beta^A)\in T^* T^* Q  \st \exists  u\in U \; \mbox{such that }\;
\alpha_i=p_A \dfrac{\partial F^A}{\partial q^i}-\dfrac{\partial L}{\partial q^i},\right. \\
&\left.\;\;\beta^A=F^A(q,u),\, p_A \dfrac{\partial F^A}{\partial u^a}-\dfrac{\partial L}{\partial u^a}=0\right\}.
\end{align*}
From the Dirac system $({D}_{\Lambda_{{\A}^*}} , S_E)$ we obtain the equations of motion
\begin{eqnarray*}
\dot{q}^j&=&\rho^j_A F^A(q, u) \, ,\\
\dot{p}_A &=& \rho^j_A\left(\dfrac{\partial L}{\partial q^j}-p_B\dfrac{\partial F^B}{\partial q^j} \right)-{\mathcal C}^C_{A B} p_C F^B(q,u)\, , \\
0&=&p_A \dfrac{\partial F^A}{\partial u^a}-\dfrac{\partial L}{\partial u^a}\,.
\end{eqnarray*}
These equations can be compared with the equations obtained for an optimal control problem defined on Lie algebroids in~\cite{2007Martinez}. When the function $E$ is a Morse family, then everything falls into the description in~\cite{2012BaIgleMar} where the integrability algorithm is used to find the solutions. In the next section the algorithm will be adapted to the family of generalized Dirac systems described in this paper.

\section{Integrability algorithm and Dirac systems}\label{sec:algorithm}

We are going to adapt the integrability algorithm developed in~\cite{1995MMT} to solve implicit differential equations to the case of Dirac systems defined by Morse families.  We first review the integrability algorithm prior to adapt it to the unified framework described in this paper.

\paragraph{Integrability algorithm.} When an implicit differential equation is given, the algorithm allows to find a subset, if it exists, where the solution curves are. The steps in the algorithm guarantee that the curves will not leave that final subset because of imposed tangency conditions. 

Let $S$ be an implicit differential equation on a manifold $P$, that  is, a submanifold $S$ of $TP$.
In such a case, it is possible to construct an algorithm to extract the integrable part of $S$ in $P$ (see ~\cite{1995MMT}).

A curve $\gamma\colon I \subseteq \mathbb{R}\rightarrow P$
is called a \textit{solution of the differential equation} $S$ if $\dot{\gamma}(I)\subset S$. The implicit differential equation $S$ is said to
be \textit{integrable at $v\in S$} if there
is a solution $\gamma\colon I \subseteq \mathbb{R}\rightarrow P$ such that
$\dot{\gamma}(0)=v$. The implicit differential equation $S$ is said to be integrable if it is integrable at each point
$v\in S$.

\begin{proposition}\cite[Proposition 5]{1995MMT} Let $\tau_P\colon TP \rightarrow P$ be the canonical tangent bundle projection.
If $N=\tau_P(S)$ is a submanifold of $P$ and if the mapping
\begin{align*}
 \tau_P\colon & S \rightarrow N,\\
&v  \mapsto \tau_P(v)
\end{align*}
is a surjective submersion, then the condition $S\subset TN$ is sufficient for integrability
of the implicit differential equation $S$.
\end{proposition}

In order to obtain the integrable part of an implicit differential equation we construct the following sequence of objects
\begin{equation*}
 (S^0,N^0,\tau^0),\; (S^1,N^1,\tau^1), \; \dots \; , \; (S^k,N^k,\tau^k), \; \dots
\end{equation*}
where
\begin{alignat*}{3}
&S^0   = S,  &\qquad\qquad  & N^0   =\tau_P(S) = N, &\qquad\qquad  &\tau^0= \tau_P,\\
& S^1  = S\cap TN,  &\qquad\qquad  & N^1  = \tau_P(S^1), &\qquad\qquad  & \tau^1=  \tau_{P|_{S^1}},\\
& \dots  &\qquad\qquad & \dots  &\qquad\qquad   & \dots \\
&S^k   = S^{k-1}\cap TN^{k-1},  &\qquad\qquad  & N^k   =\tau_P(S^k), &\qquad\qquad  &\tau^k=\tau_{P|{S^k}}.
\end{alignat*}

For each $k$, it is assumed that the sets $N^k$ are submanifolds and that the mappings $\tau^k$
are surjective submersions. Since the dimension of $P$ is finite, the sequence of implicit differential
equations $S^0$, $S^1,\dots , S^k,\dots$ stabilizes at some index $k$, that is,
$S^k=S^{k+1}$. The integrable implicit differential equation $S^k\subset TP$ is the integrable
part of $S$. We remark that $S^k$ is possibly empty.

\paragraph{Adaptation of the integrability algorithm for Dirac systems.}
As described in Section~\ref{sec:Morse}, a Dirac system over a Morse family $E$ is an implicit dynamical system given by equation~\eqref{eq:GDiracsystem}. This system is defined on the Whitney sum $TN\oplus T^*N$. We will apply the integrability algorithm in the part corresponding to the tangent bundle by defining the following submanifold of $TN$:
\begin{equation*}S_{D_N,E}=\{v\in TN \, | \, \exists \, \alpha\in S_E \; \mbox{such that } (v,\alpha)\in D_N\}.
\label{eq:SDE}
\end{equation*}
A solution to the dynamical system $(D_N,S_E)$ is a curve $n(t)$ in $N$ such that $\dot{n}(t)$ lies in $S_{D_N,E}$. To obtain the integrable part of $S_{D_N,E}$ in $N$, we start by taking
$S^0=S_{D_N,E}$, $N^0=\tau_N(S^0)$ and $\tau^0=\tau_N$.

The following steps of the algorithm are defined by 
\begin{equation*}
S^k=T\tau_N(S^{k-1})\cap S^{k-1}, \quad N^k=\tau_N(S^k), \quad \tau^k=\tau_{N|_{S^k}}\, .
\end{equation*}
If the algorithm stabilizes, there exists a final submanifold (possibly empty) satisfying $S_{k_f}=T\tau_N(S_{k_f}) \cap S_{k_f}$. The steps of the algorithm generate a sequence of submanifolds in $N$ as follows
\begin{equation*}
N_{k_f} \stackrel{i_{k_f}}{\hooklongrightarrow} N_{k_f-1} \stackrel{i_{k_{f-1}}}{\hooklongrightarrow}\dots N_1 \stackrel{i_{1}}{\hooklongrightarrow} N_0  \stackrel{i_{0}}{\hooklongrightarrow} N.
\end{equation*}
As a consequence, for every $x$ in $N_{k_f}$ there exists $v$ in $T_xN_{k_f}$ such that $v$ is in $(S_{D_N,E})_x$. In this way we have found the base submanifold where the original dynamical system has solution. Thus a solution to the dynamical system $(D_N,S_E)$ is a curve $n(t)$ on $N_{k_f}$ such that $(\dot{n}(t),\mu_{n(t)})\in (D_N)_{n(t)}$, where $\mu_{n(t)}\in (S_E)_{n(t)}$.

By the properties of Dirac structures reviewed in Section~\ref{sec:Background}, and assuming some regularity condition\footnote{For instance, if $D_N \cap (\{0\}\oplus TN_{k_f}^0)$ has constant rank. This is the so-called ``clean intersection condition'' in~\cite{Burs}.}, the backward of $D_N$ by the map  $i_{0f}=i_0\circ i_1 \circ \dots \circ i_{k_f}\colon N_{k_f} \hookrightarrow N$ defines a Dirac structure $D_{N_{k_f}}=\B_{i_{0f}}(D_N)$  on $N_{k_f}$. This condition is not necessarily always true. Even though we can make some comments about the dynamics on $D_{N_{k_f}}$ if the Morse family is also pull-backed. A solution of $(D_N,S_E)$ is always a solution of $(D_{N_{k_f}},i_{0f}^*(S_E))$, but the converse does not hold in general, as is known for the case of singular Lagrangians studied in~\cite{Dirac,1979GoNe,1978GoNeHi}.

\remark A comprehensive study of the constraint algorithm for standard Dirac systems can be found in~\cite{CEF_DiracConstraints}.

\paragraph{A preservation result.} We will now prove that if the Dirac structure is integrable, any solution of the Dirac system in the final constraint manifold $N_{k_f}$ preserves the presymplectic structure along the corresponding leaf. This result generalizes the case in which the Lagrangian submanifold is given by the graph of an energy function.

The assumption of integrability on $D_N$ defines a presympectic foliation $\mathcal{F}={\rm pr}_1(D_N)$ of $N$ where ${\rm pr}_1\colon D_N \rightarrow TN$ is the projection onto the tangent bundle. If we apply the above integrability algorithm, we obtain the final submanifold $N_{k_f}\subset N$ (we assume it is non empty). It is possible to define a presymplectic foliation $\mathcal{F}\cap T
 N_{k_f}$ of $N_{k_f}$. The leaf of the foliation $\mathcal{F}\cap TN_{k_f}$ through a point $z\in N_{k_f}$ will be denoted by $\mathcal{H}^f_z$, and the corresponding presymplectic structure on $\mathcal{H}^f_z$ will be denoted by $\omega^f_z$. Note that, by definition of $N_{k_f}$, for any $X\in\mathfrak{X}(\mathcal{H}^f_z)$ there exists a 1-form $\mu$ in $S_E$ defined along $\mathcal{H}^f_z$  such that 
\[
(X,\mu)\in D_{N_y} \quad \mbox{for any } y\in \mathcal{H}^f_z\, .
\]
Consequently, $i_X\omega^f_z=\mu$ on $\mathcal{H}^f_z$ pointwise. We are now ready to prove the following result:
\begin{proposition}
 Let $X$ be a vector field on $\mathcal{H}^f_z$ such that there exists $\mu$ in $S_E$ defined along $\mathcal{H}^f_z$ satisfying $i_X\omega^f_z=\mu$. Then $\pounds_X\omega^f_z=0$. In particular, the flow of $X$ preserves $\omega^f_z$. 
\end{proposition}
\begin{proof}
 Using Cartan's formula and the definition of $D_{N_z}$ we have $\pounds_X\omega^f_z=d(i_X\omega^f_z)=d\mu$. As  $S_E$ is a Lagrangian submanifold of $T^*N$, for any $Y_1,Y_2\in \mathfrak{X}(S_E)$  we have that $0=\omega_N(Y_1,Y_2)=-d\theta_N(Y_1,Y_2)$, where $\theta_N$ is the canonical 1-form on $T^*N$. At $\nu\in T^*N$ we have
 \begin{align*}
  0=-d\theta_N(Y_1,Y_2)\mid_\nu&=\left(-Y_1[\theta_N(Y_2)]+Y_2[\theta_N(Y_1)]+\theta_N([Y_1,Y_2])\right)\mid_\nu \\
  &=-Y_1[\nu(T\pi_N(Y_2))]+Y_2[\nu(T\pi_N(Y_1))]+\nu([T\pi_N(Y_1),T\pi_N(Y_2)]).
 \end{align*}
In particular, if $\nu$ is understood as a 1-form on $N$, the previous computation for $Y_1,Y_2\in \mathfrak{X}(S_E)$ can be rewritten as follows:
\[
 0=-d\theta_N(Y_1,Y_2)\mid_\nu=-(d\pi_N^*\nu)(Y_1,Y_2)=-(\pi_N^*d\nu)(Y_1,Y_2)=-d\nu(T\pi_N(Y_1),T\pi_N(Y_2)).
\]
But this implies that for any $Z_1,Z_2\in\mathfrak{X}\in(\mathcal{H}^f_z)$ obtained from the projections of vector fields on $S_E$ we have  

\[
 (\pounds_X\omega^f_z)(Z_1,Z_2)=d\mu(Z_1,Z_2)=0.
\]
It remains to prove that any vector field $Z$ on $\mathcal{H}^f_z$ can be expressed as a projection of  a vector field on $S_E$. For any integral curve $n(t)$ of $Z$ in the leaf $\mathcal{H}^f_z$, the integrability algorithm guarantees that there exists $\mu_{n(t)}$ in $(S_E)_{n(t)}$ such that $(\dot{n}(t),\mu_{n(t)})\in (D_N)_{n(t)}$. As $\mu_{n(t)}=(i_{\dot{n}(t)}\omega^f_z)_{n(t)}$, the differentiability of $\mu$ guarantees that the result holds. Note that, as needed,
\begin{equation*}
T\pi_N\left(\dfrac{{\rm d}}{{\rm d}t} \mu_{n(t)}\right)=\dot{n}(t)=Z(n(t))\, .
\end{equation*}
\end{proof}

\section{Future work}\label{sec:future}

The unified formalism developed in this paper includes and extends recent results in the literature such as~\cite{2015JiYo} where nonholonomic and vakonomic mechanics are described using Dirac structures. As the authors prove there, nonholonomic mechanics cannot be described by Lagrangian submanifolds. However, a Morse family also defines a Lagrangian submanifold that can be used to provide the Dirac structures with dynamics. This approach makes possible to describe nonholonomic mechanics as described in Section~\ref{subsec:almost}, as well as many other examples explained in Section~\ref{sec:Examples}.

The future research lines include:

\begin{enumerate}
\item The interconnection of simpler systems allows to describe complex systems  as the theory of port-Hamiltonian systems shows~\cite{VdS-Book,VdSM2}. Dirac structures have already been used in this context~\cite{CerVdSBan,CerVdSBan2} (we also refer to \cite{2018_Interconnection}  for a recent geometric approach), but our formalism provides an intrinsic description to tackle the interconnection of Dirac systems by means of Lagrangian submanifolds and Morse families. 

 \item To solve the generalized Dirac systems is usually challenging. This paper reveals a unified geometric approach that opens the path to define geometric integrators for those Dirac systems~\cite{HLW_book,LeOh,MW_Acta,MR3707341}.

\end{enumerate} 

\appendix

\section{Natural maps and commutative diagrams}\label{ap:A}

There are some interesting relations between the Dirac structure $D_{\omega_Q}$ on $T^*Q$, given by the graph of $\omega_Q$, and the spaces $TT^*Q$, $T^*T^*Q$ and $T^*TQ$ appearing in the Tulczyjew triple. Let us first recall that there are canonically defined isomorphisms whose local expressions are
\begin{alignat*}{3}
\flat_{\omega_Q}\colon TT^*Q&\to T^*T^*Q\,,\qquad (q^i,p_i,\dot q^i,\dot p_i)&&\mapsto (q^i,p_i,-\dot p_i,\dot q^i),\\
\alpha_Q\colon TT^*Q&\to T^*TQ\,,\qquad (q^i,p_i,\dot q^i,\dot p_i)&&\mapsto (q^i,\dot q^i,\dot p_i,p_i).
\end{alignat*}
The intrinsic definition of these maps, and the motivation behind them, can be found in~\cite{Tu,TuHamilton}. With these maps in mind, and noting that $D_{\omega_Q}$ reads
\begin{equation}
D_{\omega_Q} =\left\{(q^i,p_i,\dot q^i,\dot p_i,\alpha_i,\beta^i)\st \dot p_i +\alpha_i=0,\, \dot q^i-\beta^i=0\right\},
\end{equation}
we will define diffeomorphisms $\Psi_1$, $\Psi_2$ and $\Psi_3$ from $D_{\omega_Q}$ to the spaces $TT^*Q$, $T^*TQ$ and $T^*T^*Q$, respectively. Namely, if we denote by $\hbox{pr}_1$ and $\hbox{pr}_2$ the projections of $TT^*Q\oplus T^*T^*Q$ onto $TT^*Q$ and $T^*T^*Q$, and by $i_{D_{\omega_Q}}: D_{\omega_Q}\hookrightarrow TT^*Q\oplus T^*T^*Q$. Then we set:
\begin{alignat*}{2}
&\Psi_1\colon D_{\omega_Q}\subset TT^*Q\oplus T^*T^*Q \to TT^*Q \,, \qquad &&\Psi_1  =\hbox{pr}_1\circ i_{D_{\omega_Q}},\\
&\Psi_2\colon D_{\omega_Q}\subset TT^*Q\oplus T^*T^*Q \to T^*T^*Q \,, \qquad &&\Psi_2 =\hbox{pr}_2\circ i_{D_{\omega_Q}},\\
&\Psi_3\colon D_{\omega_Q}\subset TT^*Q\oplus T^*T^*Q \to T^*T^*Q \,, \qquad && \Psi_3 = \alpha_Q\circ\Psi_1.
\end{alignat*}

\noindent These maps and their coordinate expressions are shown in Diagram~\ref{dia:Psi} (where $\sharp_{\omega_Q}$ is the inverse of $\flat_{\omega_Q}$).

\begin{figure}[ht!]
\centering
\includegraphics{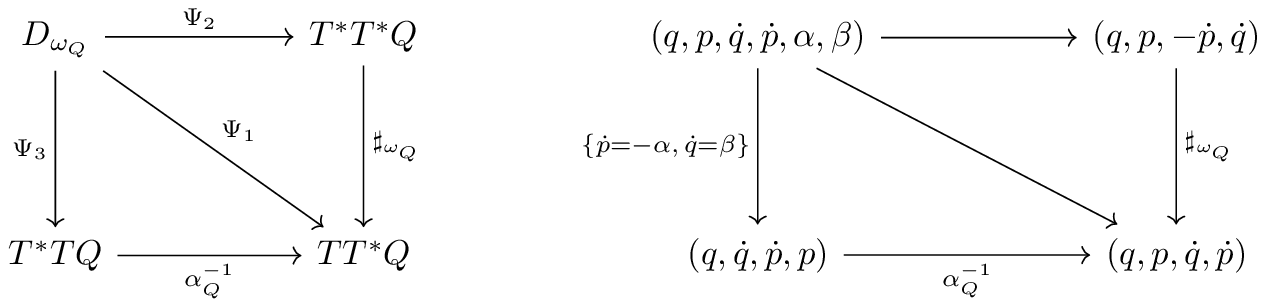}
\caption{The maps $\Psi_1$, $\Psi_2$ and $\Psi_3$.}\label{dia:Psi}
\end{figure}

One might also relate $D_{\omega_Q}$ to the Dirac structure $D_M$ on $M=TQ\oplus T^*Q$, given by the graph of the pullback of $\omega_Q$ to $M$. In coordinates
\begin{equation}\label{eq:D_M}
D_{M} =\left\{(q^i,v^i,p_i,\dot q^i,\dot v^i,\dot p_i,\alpha_i,\gamma_i,\beta^i)\st \dot p_i +\alpha_i=0,\, \gamma_i=0,\, \dot q^i-\beta^i=0\right\}.
\end{equation}
There are natural maps
\begin{figure}[H]
\centering
\includegraphics{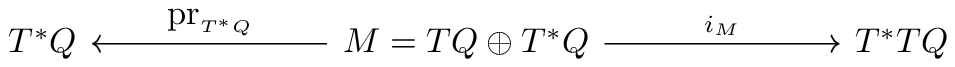}
\end{figure}

\noindent where $\mbox{pr}_{T^*Q}$ is the projection onto $T^*Q$, and $i_M\colon M\to  T^*TQ$ is the inclusion (it is a vector bundle inclusion, see~\cite{CEF_DiracConstraints} for a definition). In coordinates, $i_M(q,v,p)=(q,v,p,0)$. The basic observation is that the spaces $T^*Q $ and $T^*TQ$ both have canonical Dirac structures $D_{\omega_Q}$ and $D_{\omega_{TQ}}$, and that $D_M$ can be obtained via the backward of these structures by either the projection or the inclusion, i.e. $D_M=\B_{\mbox{pr}_{T^*Q}}(D_{\omega_Q})=\B_{i_M}(D_{\omega_{TQ}})$. We summarize the situation in an enlarged diagram (Diagram~\ref{dia:D_M}).
\begin{figure}[ht!]
\centering
\includegraphics{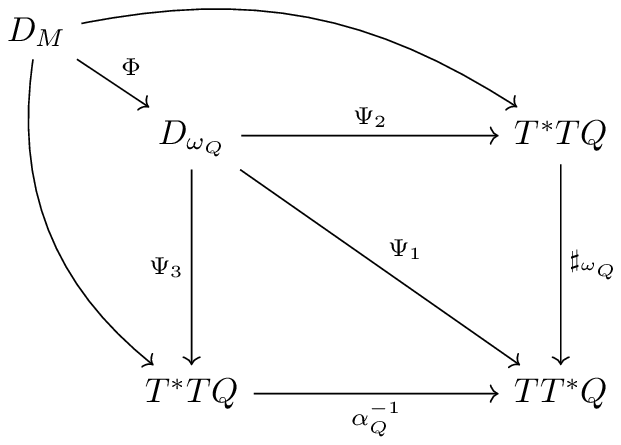}
\caption{$D_M$ and $D_{\omega_Q}$.}\label{dia:D_M}
\end{figure}

\noindent The map $\Phi\colon D_M\to D_{\omega_Q}$ is a vector bundle morphism over the projection $\mbox{pr}_{T^*Q}$, in coordinates
\begin{equation}\label{eq:Phi}
\Phi(q^i,v^i,p_i,\dot q^i,\dot v^i,\dot p_i,\alpha_i,\gamma_i,\beta^i)=(q^i,p_i,\dot q^i,\dot p_i,\alpha_i,\beta^i).
\end{equation}
A coordinate-free definition is $ \Phi=T{\mbox{pr}_{T^*Q}}\oplus T^*i_{T^*Q}$, where $i_{T^*Q}\colon T^*Q\to M=TQ\oplus T^*Q$ is the inclusion $i_{T^*Q}(q,p)=(q,0,p)$, and $T^*i_{T^*Q}$ is the cotangent map of $i_{T^*Q}$.

In the case of a linear almost Poisson bracket discussed in Section~\ref{subsec:almost}, there is a similar commutative diagram for $D_{\A^*}$ (compare with Diagram~\ref{dia:Psi}):

\begin{figure}[H]
\centering
\includegraphics{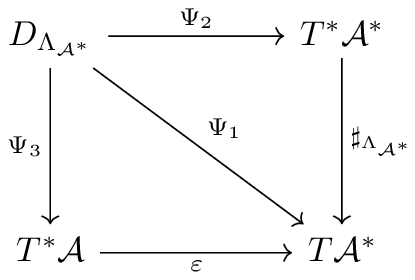}
\end{figure}

\noindent To define the maps involved, let us first recall the existence of a canonical isomorphism ${\mathcal{R}}\colon T^*\A^*\to T^*\A$, in coordinates
\[
{\mathcal{R}}(q^i,p_A,\alpha_i,\beta^A)=(q^i,\beta^A,-\alpha_i,p_A).
\]
The map $\varepsilon$ in the diagram is such that the following diagram is commutative
\begin{figure}[H]
\centering
\includegraphics{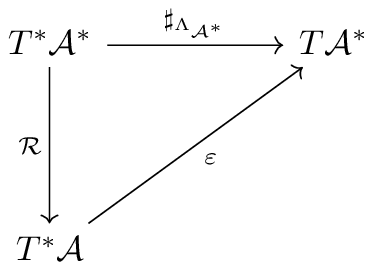}
\end{figure}

\noindent (in particular, it depends on the Poisson structure chosen). In coordinates, we find
\[
\varepsilon(q^i,v^A,\alpha_i,\gamma_A)=(q^i,\gamma_A,\rho_A^i v^A,C^{D}_{BA} v^{B}\gamma_{D} -  \rho^i_{A}\alpha_i).
\]
We refer the reader to~\cite{GrGrUr,GrUr_Algebroids} for more details. We can now define the maps $\Psi_1$, $\Psi_2$ and $\Psi_3$. Denote by $\hbox{pr}_1$ and $\hbox{pr}_2$ the projections of $T\A^*\oplus T^*\A^*$ onto $T^*\A^*$ and $T\A^*$, and $i_{D_{\Lambda_{\A^*}}}\colon D_{\Lambda_{\A^*}}\to T\A^*\oplus T^*\A^*$ the inclusion. Then:
\begin{alignat*}{2}
&\Psi_1\colon D_{\Lambda_{\A^*}} \to T\A^* \,, \qquad &&\Psi_1  =\hbox{pr}_1\circ i_{D_{\Lambda_{\A^*}}},\\
&\Psi_2\colon D_{\Lambda_{\A^*}} \to T^*\A^* \,, \qquad &&\Psi_2 =\hbox{pr}_2\circ i_{D_{\Lambda_{\A^*}}},\\
&\Psi_3\colon D_{\Lambda_{\A^*}} \to T^*\A \,, \qquad && \Psi_3 = {\mathcal{R}}\circ \hbox{pr}_2\circ i_{D_{\Lambda_{\A^*}}}={\mathcal{R}}\circ \Psi_2.
\end{alignat*}
The coordinate expressions are summarized in the following diagram:

\begin{figure}[H]
\centering
\includegraphics{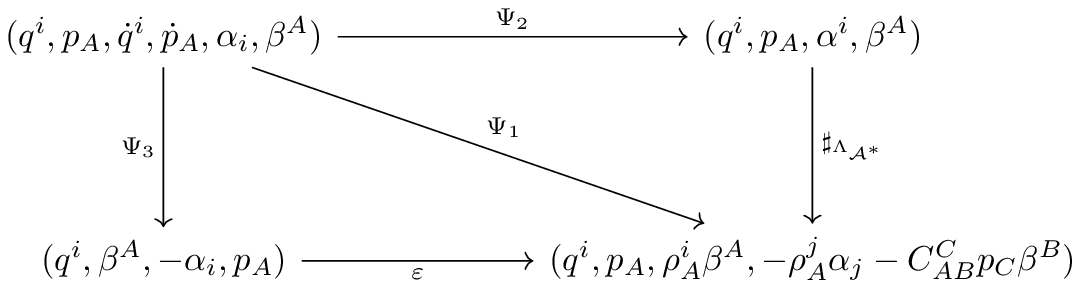}
\end{figure}

We also have an enlarged commutative diagram generalizing Diagram~\ref{dia:D_M}:

\begin{figure}[H]
\centering
\includegraphics{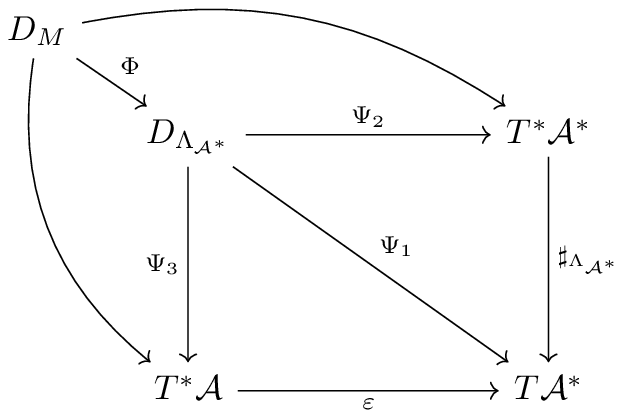}
\end{figure}

\noindent where $D_M=\B_{\pi_{(M,\A^*)}}\left(D_{\Lambda_{\A^*}}\right)$ and the map $\Phi$ is defined analogously to~\eqref{eq:Phi}.

\bibliographystyle{alpha}
\bibliography{References}

\end{document}